\RequirePackage{fix-cm}
\documentclass[smallextended]{svjour3}       
\smartqed  

\usepackage{graphicx}
\usepackage{graphicx}
\usepackage{empheq}
\usepackage{subfigure}

\usepackage{tikz-cd}
\usepackage{tikz}

\usepackage[numbers]{natbib}
\usepackage[integrals]{wasysym}
\usepackage{changes}
\usetikzlibrary{decorations.markings}
\usetikzlibrary{shapes.geometric, arrows}
\usepackage{tikzscale}
\usepackage{filecontents}
\usepackage{wrapfig}
\usepackage{tcolorbox}
\usepackage{color}
\usepackage{siunitx}
\definecolor{blue1}{rgb}{0.0, 0.0, 1.0}
\definecolor{gray}{rgb}{0.9,0.9,0.9}
\definecolor{gray1}{rgb}{0.7,0.7,0.7}
\definecolor{gray2}{rgb}{0.8,0.8,0.8}
\definecolor{magenta}{rgb}{1.0, 0.0, 1.0}
\usepackage[title]{appendix}
\newcommand\backmatter{\appendix
\def\chaptermark##1{\markboth{%
\ifnum  \c@secnumdepth > \m@ne  \@chapapp\ \thechapter:  \fi  ##1}{%
\ifnum  \c@secnumdepth > \m@ne  \@chapapp\ \thechapter:  \fi  ##1}}%
\def\sectionmark##1{\relax}}

\newcommand{\ISI}[1]{\mathrm{ISI}}

\usepackage{hyperref}
\hypersetup{ colorlinks=true, urlcolor  = blue, linkcolor = blue,
citecolor = blue1,}

\usepackage{float}
\usepackage{amsmath,mathtools,amssymb,amsfonts}
\usepackage[scaled=.95]{helvet}

\definecolor{lime}{HTML}{A6CE39}
\DeclareRobustCommand{\orcidicon}{%
    \begin{tikzpicture}
    \draw[lime, fill=lime] (0,0) 
    circle [radius=0.16] 
    node[white] {{\fontfamily{qag}\selectfont \tiny ID}};
    \draw[white, fill=white] (-0.0625,0.095) 
    circle [radius=0.007];
    \end{tikzpicture}
    \hspace{-2mm}
}

\foreach \x in {A, ..., Z}{%
    \expandafter\xdef\csname orcid\x\endcsname{\noexpand\href{https://orcid.org/\csname orcidauthor\x\endcsname}{\noexpand\orcidicon}}
}
\newcommand{\orcid}[1]{\href{https://orcid.org/#1}{\textcolor[HTML]{A6CE39}{\orcidicon}}}

\newcommand{\uproman}[1]{\uppercase\expandafter{\romannumeral#1}}

 \journalname{Journal}

\begin{document}
\title{Reduced-order adaptive synchronization in a chaotic neural network with parameter mismatch: A dynamical system vs. machine learning approach}

\titlerunning{Reduced-order adaptive synchronization in a chaotic neural network}        

\author{Jan Kobiolka, Marius E. Yamakou\orcid{0000-0002-2809-1739}}


\institute{
J. Kobiolka, M. E. Yamakou  \at  Department of Data Science, Friedrich-Alexander-Universit\"at Erlangen-Nürnberg, Cauerstr. 11, 91058 Erlangen, Germany.\\
             \email{marius.yamakou@fau.de}\\  
}

\date{Received: date }

\maketitle
\begin{abstract}
In this paper, we address the reduced-order synchronization problem between two chaotic memristive Hindmarsh-Rose (HR) neurons of different orders using two distinct methods. The first method employs the Lyapunov active control technique. Through this technique, we develop appropriate control functions to synchronize a 4D chaotic HR neuron (response system) with the canonical projection of a 5D chaotic HR neuron (drive system). Numerical simulations are provided to demonstrate the effectiveness of this approach. The second method is data-driven and leverages a machine learning-based control technique. Our technique utilizes an \textit{ad hoc} combination of reservoir computing (RC) algorithms, incorporating reservoir observer (RO), online control (OC), and online predictive control (OPC) algorithms. We anticipate our effective heuristic RC adaptive control algorithm to guide the development of more formally structured and systematic, data-driven RC control approaches to chaotic synchronization problems, and to inspire more data-driven neuromorphic methods for controlling and achieving synchronization in chaotic neural networks \textit{in vivo}.
 
\end{abstract}

\section{Introduction}
Chaotic behavior is a fascinating occurrence that manifests in nonlinear systems and has garnered increasing interest in scholarly articles over the past thirty years. A chaotic system, which is nonlinear and deterministic, exhibits intricate and unpredictable behavior. The system's sensitivity to initial conditions and parameter variations are key features, making the challenge of achieving chaotic synchronization significantly important. The synchronization of chaotic systems has been mainly categorized into complete synchronization \cite{pecora1990synchronization,ma2011complete,yamakou2016ratcheting}, generalized synchronization \cite{kocarev1996generalized,yang2007three}, lag synchronization \cite{rosenblum1997phase,wang2011lag,wang2009chaotic}, anticipated synchronization \cite{calvo2004anticipated}, phase synchronization \cite{rosenblum1997phase,belykh2005automatic},
and practical synchronization \cite{femat1999chaos}. In chaotic systems, it has been shown that time delays can trigger transitions between some of these forms of synchronization \cite{shu2005switching}. 
The literature on dynamics and control of these types of synchronization is abundant, see, e.g., \cite{boccaletti2002synchronization,arenas2008synchronization,tang2014synchronization} and the references therein.

Chaos synchronization has attracted significant attention since Pecora and Carroll's seminal paper \cite{pecora1990synchronization}. Consequently, the synchronization of chaotic systems has been extensively explored due to its promising applications across multiple domains, including neuro\-sci\-ence \cite{guevara2017neural,yamakou2023synchronization}, power grids \cite{totz2020control}, secure communication \cite{kocarev1995general}, electronic circuits \cite{ma2010time}, and others. Numerous chaos synchronization techniques have been devised, including drive-response control \cite{yang1999note}, coupling control \cite{lu2002chaos}, variable structure control \cite{yin2002synchronization}, impulsive control \cite{wang2004impulsive,sun2002impulsive}, active control \cite{ho2002synchronization,yassen2005chaos}, and adaptive control \cite{han2004adaptive,liao2000adaptive,chen2002synchronization} 
which can be categorized into Lyapunov-based \cite{liao2000adaptive,vincent2009simple} and identification-based approaches \cite{krstic1995nonlinear}, depending on the parameter update law and the associated stability and convergence proofs. The identification-based approach offers greater flexibility in selecting an appropriate identifier module, allowing it to be treated as entirely distinct from the controller module. Conversely, the Lyapunov-based approach (utilized in this paper) relies on the Lyapunov stability theory to establish stability. However, a significant challenge with this method is identifying a suitable Lyapunov function.

Nonetheless, the control methods listed earlier and other control techniques primarily address the synchronization of two identical chaotic systems with either known parameters or identical unknown parameters. In practical scenarios, it is rare for the drive and response chaotic systems to have identical structures. Furthermore, the system parameters are invariably affected by external factors and are not precisely known beforehand. Hence, synchronizing two distinct chaotic systems in the presence of unknown parameters is crucial and highly beneficial in real-world applications.

Moreover, in control scenarios involving coupled systems with unknown parameters \cite{chen2002synchronization,xiong2023weak}, additional difficulties may arise in synchronizing these systems if they are not of the same order, i.e., they have different phase space dimensions) \cite{motallebzadeh2012synchronization,zhu2008full,ogunjo2013increased}.  In same-order (same phase space dimension) synchronization problems, the drive and response systems share similar geometric and topological characteristics \cite{yang1999note,chen2002synchronization,hu2005adaptive}. Consequently, a straightforward drive-response connection with a suitable coupling is typically sufficient for synchronization. Synchronizing different-order chaotic systems \cite{motallebzadeh2012synchronization,jouini2019increased,femat2002synchronization} presents a challenge due to several factors: (i) the initial conditions of the drive and response systems might be different or even unknown; (ii) the topological and geometrical properties of the two chaotic systems vary considerably; and (iii) the presence of partially or fully unknown parameters, or parameter mismatches/uncertainties between the systems can lead to entirely distinct time evolution.

As a result, the literature has predominantly focused on chaos synchronization in systems of the same dimension, and relatively fewer research works have focused on different-order synchronization problems, i.e., reduced-order synchronization \cite{femat2002synchronization,bowong2006adaptive,bowong2004stability} or increased-order synchronization \cite{ogunjo2013increased,qing2009increasing,al2011adaptive}.  There are practical scenarios where different-order systems need to be synchronized, such as the synchronization between the heart and lungs, the coordination between thalamic and hippocampal neurons, and the alignment of neuron systems with certain biomechanical systems (e.g., biological implants). Thus, studying the synchronization of different-order chaotic systems is crucial for both practical applications and control theory.

Some real-world systems are too complex to formulate the dynamical equations that govern their behaviors. For instance, there is no comprehensive equation for the brain or climate, and often the most viable approach is to rely on the data obtained through direct measurement of these complex systems \cite{brunton2022data}. This is where data-driven methods and  machine learning (ML) can come to our rescue. By leveraging ML, we can model and control complex systems in ways that would be otherwise impossible with traditional analytical methods. For example, ML can create accurate models of chaotic systems based on observational data. Techniques such as recurrent neural networks (RNNs)  can learn intricate patterns and dependencies within data, enabling the replication and prediction of chaotic behaviors without explicit mathematical formulations (dynamical equations) governing these complex systems \cite{brunton2022data,chen2023synchronization}. 

Moreover, ML can be employed in the predictive control of coupled chaotic systems. For instance, ML models like RNNs can forecast future states of chaotic systems over short periods, allowing for real-time control adjustments to achieve a desired synchronization \cite{chen2023synchronization,kent2024controlling}. Furthermore, adaptive control strategies using ML can refine control mechanisms based on real-time system feedback, enabling dynamic adjustments for evolving chaotic systems \cite{brunton2022data,nazerian2023synchronizing,platt2021forecasting}. In particular, reservoir computing (RC) \cite{jaeger2004harnessing,maass2002real,gottwald2021combining,lukovsevivcius2009reservoir} is a machine learning technique for forecasting time-series data with RNNs, where only the output layer is trained. RC has proven effective in modeling complex dynamical systems, such as reconstructing attractors \cite{lu2018attractor,grigoryeva2023learning}, computing Lyapunov exponents \cite{pathak2017using}, forecasting novel attractors \cite{rohm2021model}, and to infer bifurcation diagrams \cite{kim2021teaching}. Recent applications include enhancing climate models \cite{arcomano2022hybrid}, inferring network connections \cite{banerjee2021machine}, predicting synchronization \cite{weng2019synchronization,nazerian2023synchronizing,fan2021anticipating}, forecasting tipping points \cite{kong2021machine,patel2023using}. Following its initial success in predicting chaotic systems \cite{pathak2018model}, research has focused on understanding RC's strengths/ limitations \cite{carroll2019network,jiang2019model} and has led to advancements like hybrid \cite{pathak2018model}, parallel \cite{srinivasan2022parallel}, reservoir observer \cite{lu2017reservoir}, and even deep RC \cite{gallicchio2017deep} architectures.

In this paper, we are interested in a reduced-order chaotic synchronization problem. The objective is to build (i) a Lyapunov-based active feedback control and (ii) a machine learning control based on an \textit{ad hoc} combination of reservoir computing architectures and algorithms that enable the behavior of a 4D chaotic Hind\-marsh-Rose (HR) neuron with unknown parameters to synchronize with the canonical projection of a 5D chaotic HR neuron.

The rest of the paper is organized as follows: In Section \ref{sec:HR}, we present the mathematical equations and the complexity of the 5D HR neuron model. In Sections \ref{sec:DS} and \ref{sec:ML} we study the reduced-order chaotic synchronization problem between the 4D and the 5D HR neuron via the Lyapunov-based adaptive feedback control technique and the ML technique based on RC, respectively. In Section \ref{sec:conclusion}, we summarize and present our conclusions.

\section{Mathematical model and complexity}\label{sec:HR}

\subsection{Mathematical model}
We consider a paradigmatic model with well-known biological relevance, the 5-dimensional Hindmarsh-Rose (HR) neuron model 
\cite{hindmarsh1984model,selverston2000reliable,lv2016model,yamakou2020chaotic}:
\begin{eqnarray}
\left\{\begin{array}{lcl}\label{eq:master}
\displaystyle{\frac{dx}{dt}} &=& ax^2 - bx^3 + y - z - k_{1}(\alpha+3\beta \phi^2)x + I,\\[2.0mm]
\displaystyle{\frac{dy}{dt}} &=&   c - dx^2 - y - \sigma w, \\[2.0mm]
\displaystyle{\frac{dz}{dt}} &=&   \theta[s(x - x_0) - z],\\[2.0mm]
\displaystyle{\frac{dw}{dt}} &=&   \mu[\gamma(y - y_0) - \rho w],\\[2.0mm]
\displaystyle{\frac{d\phi}{dt}} &=&  x - k_{2}\phi,
\end{array}\right.
\end{eqnarray}
where the dimensionless variables are given by the time $t\in\mathbb{R}$ and $[x,y,z,w,\phi]=[x(t),y(t),z(t),w(t),\phi(t)]\in\mathbb{R}^5$, representing a membrane potential $x$, recovery current $y$ associated with fast ions, slow adaption current $z$ associated with slow ions, slow adaption current $w$ associated with slower ions, and the magnetic flux $\phi$ across the neuron's cell membrane, respectively.

Throughout this study, some of the model's parameters are fixed at $a = 3.0$, $b = 1.0$, $\alpha = 0.1$, $\beta = 0.02$, $c = 1.0$, $d = 5.0$, $\sigma = 0.0278$, $\theta = 0.006$,  $x_0 = -1.56$, $y_0 = -1.619$, $\mu = 0.0009$, $\gamma = 3.0$, $\rho = 0.9573$, $I = 3.1$.
The remaining parameters, namely, $s\in[3.0,5.0]$ and the memristive gain parameters $k_1\in[0,5]$ and $k_2\in[0,2]$, will serve as the control/bifurcation parameters. Biologically, the parameter $k_1$  bridges the coupling and modulation on membrane potential $x$ from magnetic field $\phi$, and the parameters $k_2$ describe the degree of polarization and magnetization by adjusting the saturation of magnetic flux \cite{ma2017mode}. These parameters would be adjusted within the specified intervals to change qualitatively the dynamics of the model via bifurcations.

\subsection{Numerical bifurcation analysis and firing patterns}\label{sec:NB}
We now investigate the dynamics of the 5D HR neuron model described by Eq. \eqref{eq:master}. With initial conditions $[x(0), y(0), z(0), w(0), \phi(0)] = [0.1, 0.2, 0.3, 0.1, 0.2]$, we integrate these ordinary differential equations (ODEs). This is done using the fourth-order Runge-Kutta algorithm \cite{butcher1987numerical} with an integration time interval sufficiently long to overcome transient solutions.

Figures \ref{fig:bursting}\textbf{(a)}-\textbf{(e)} show the time series of the variables of the 5D HR neuron, highlighting five distinct firing patterns as the memristive gain parameters are varied. Throughout the rest of the paper, we fix the parameter $s$ at $s=4.75$, unless otherwise specified. Specifically, Figs. \ref{fig:bursting}\textbf{(a)}-\textbf{(e)} depict sub-threshold spiking ($k_1=2.3$, $k_2=0.5$), supra-threshold spiking ($k_1=0.1$, $k_2=0.1$), sub-threshold bursting ($k_1=5.0$, $k_2=1.5$), supra-threshold bursting ($k_1=0.08$, $k_2=0.4$), and no firing activity i.e., the neuron is in the quiescent state ($k_1=2.5$, $k_2=0.5$), respectively. The distinction between sub-threshold and supra-threshold firing patterns is based on whether the membrane potential variable $x$ oscillates below or above the (arbitrary) threshold value of $x_\text{th}=0.75$ (indicated in Fig.\ref{fig:bursting}\textbf{(a)}-\textbf{(e)} by the dashed green horizontal line), respectively. 
\begin{figure}[H]
    \centering
        \includegraphics[width=10.0cm,height=4.5cm]{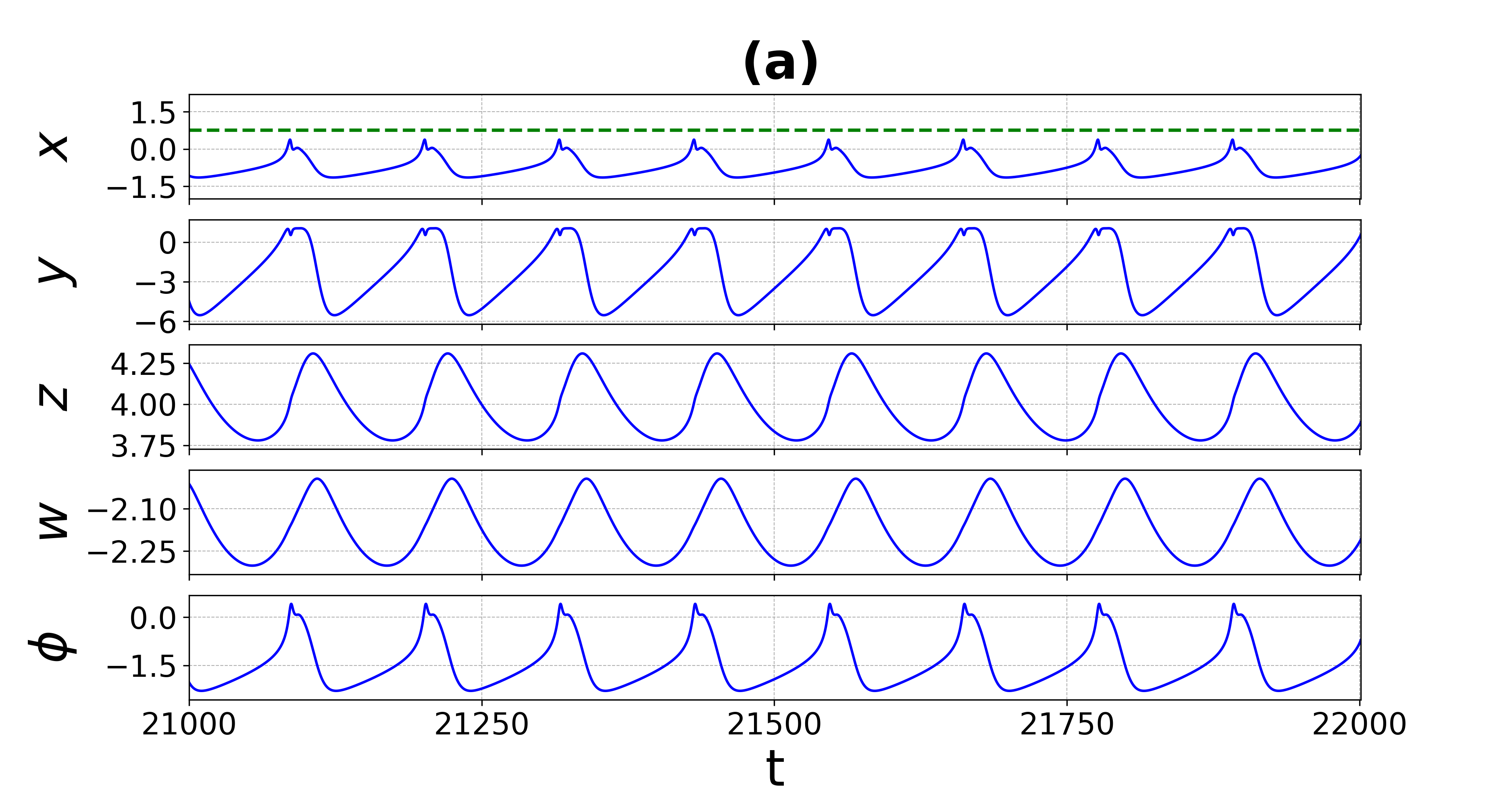}
        \includegraphics[width=10.0cm,height=4.5cm]{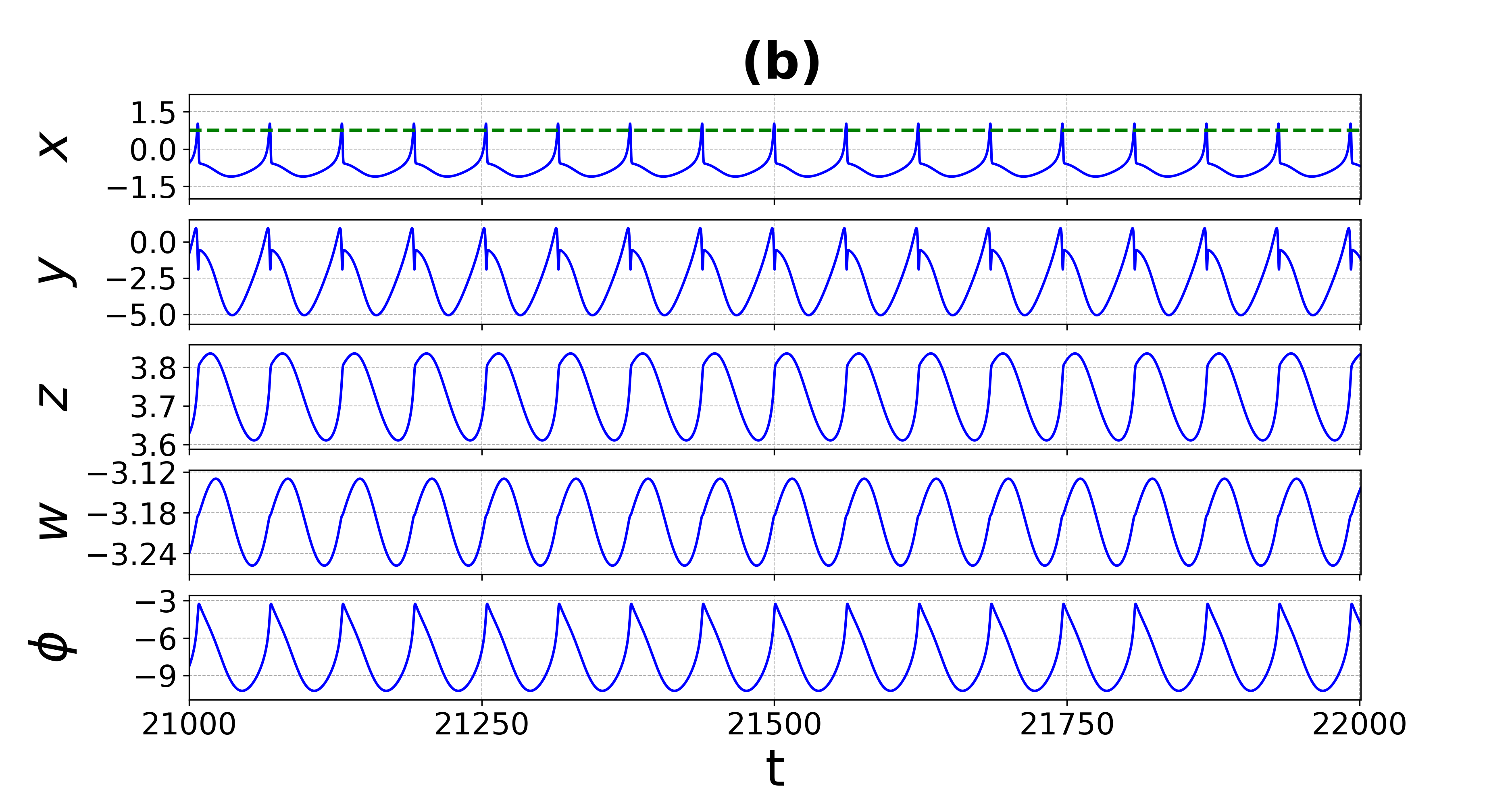}
         \includegraphics[width=10.0cm,height=4.5cm]{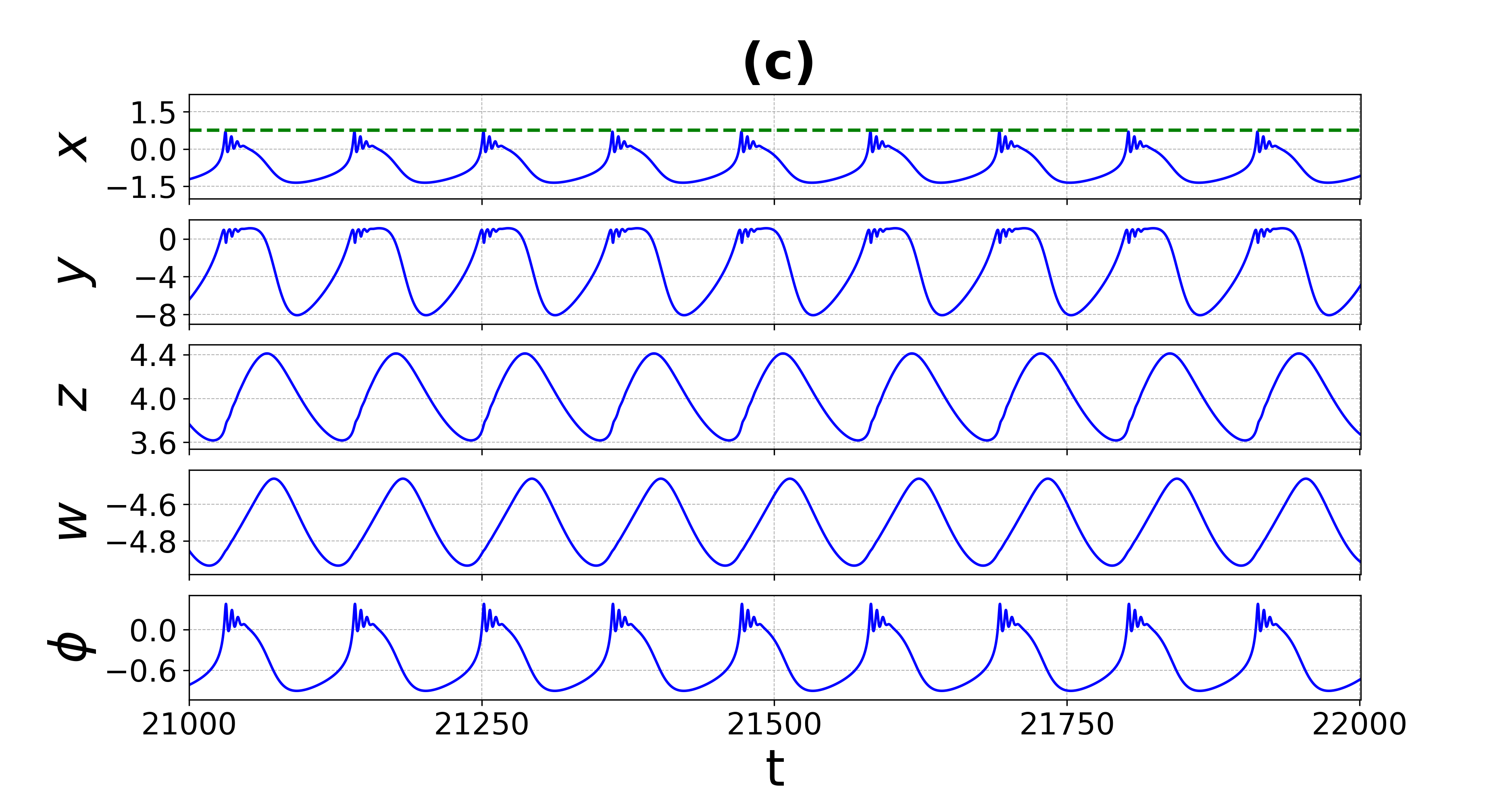}
        \includegraphics[width=10.0cm,height=4.5cm]{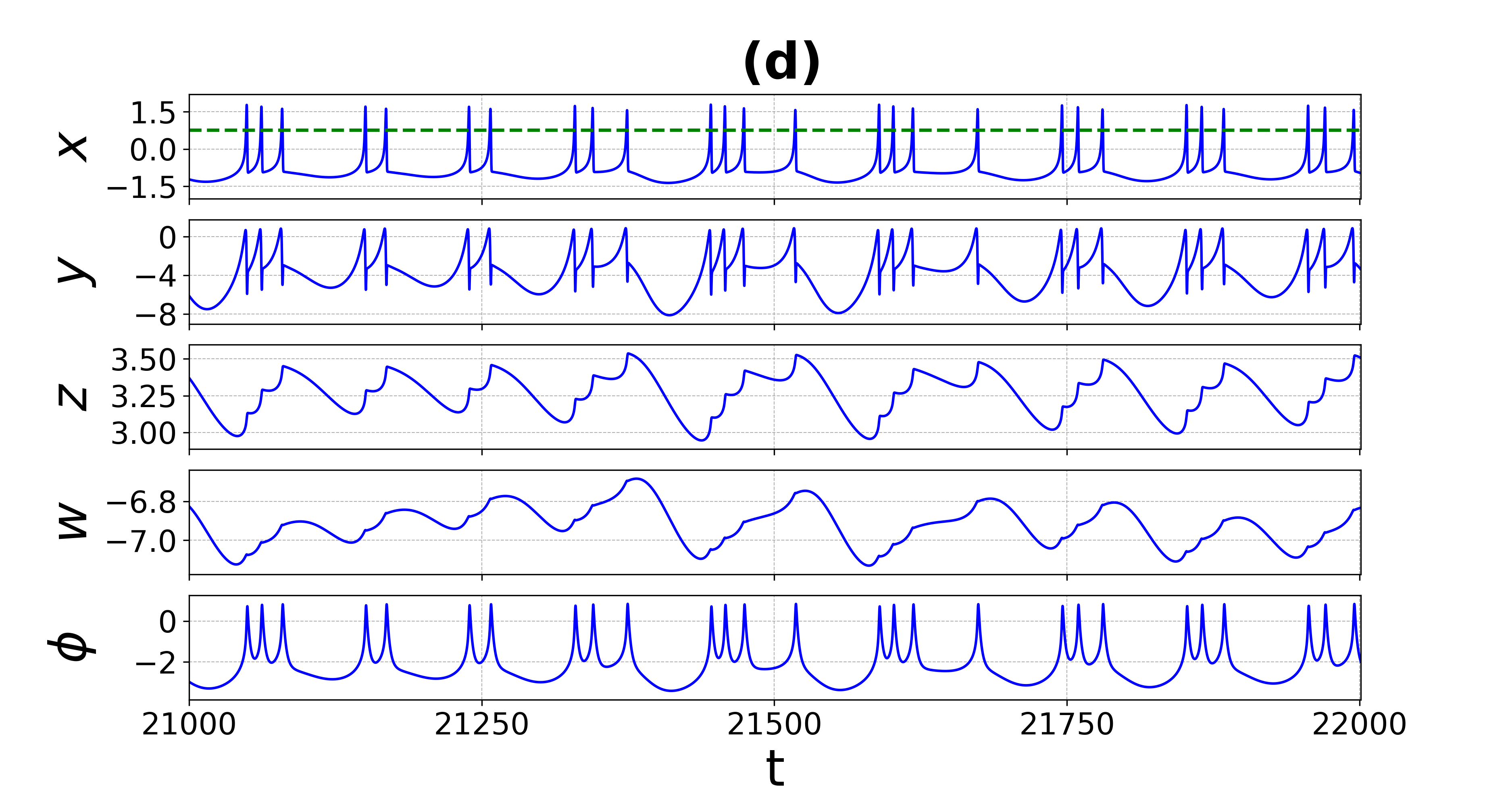}
        \includegraphics[width=10.0cm,height=4.5cm]{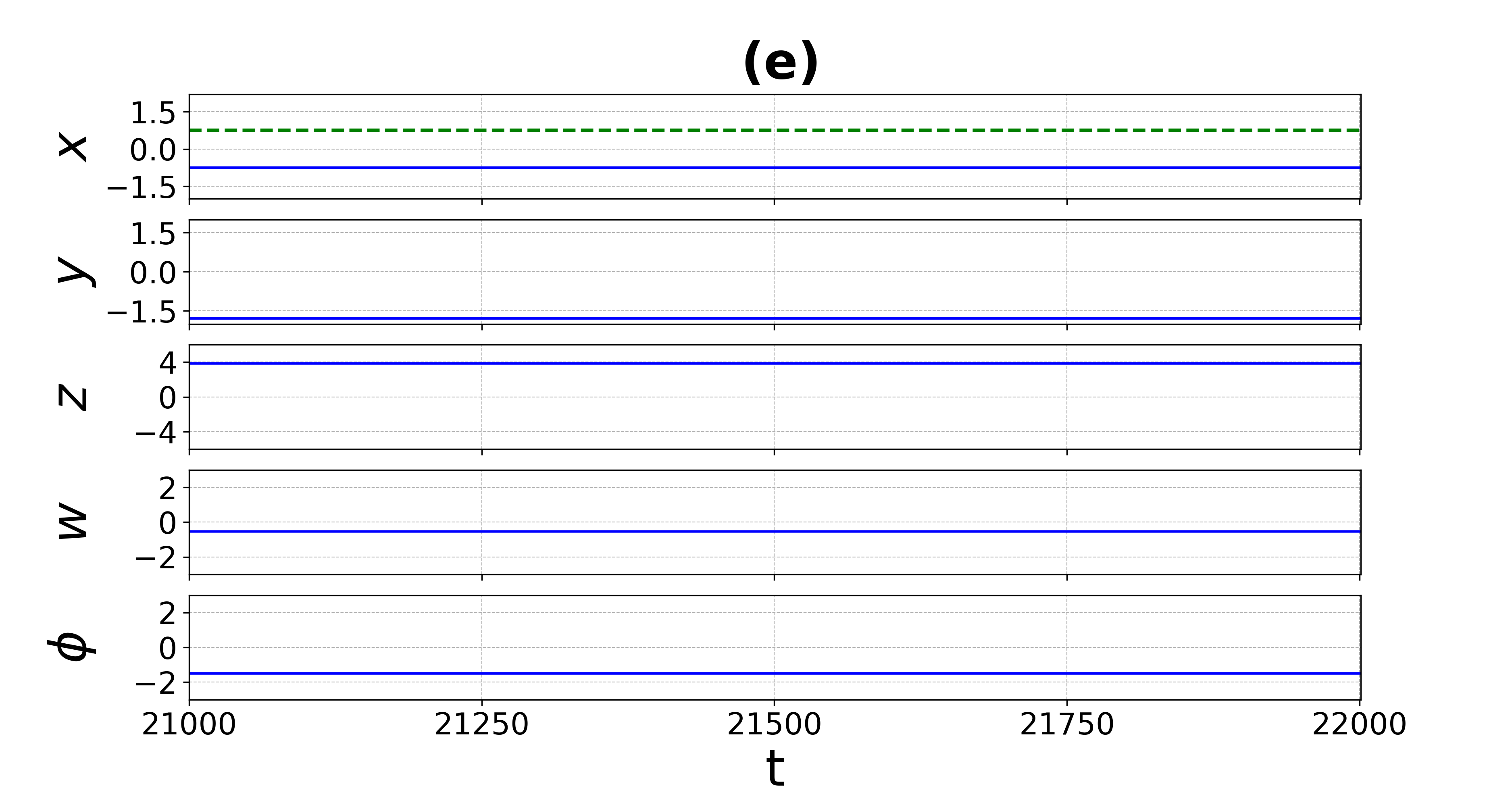}
    \caption{Activity of the 5D HR neuron model variables ($x,y,z,w,\phi$) for $s=4.75$, showing firing patterns including: \textbf{(a)} sub-threshold spiking, \textbf{(b)} supra-threshold spiking, \textbf{(c)} sub-threshold bursting, \textbf{(d)} supra-threshold bursting, and \textbf{(e)} no firing activity. The green horizontal dashed line at $x=x_\text{th}=0.75$ determine whether a firing pattern is sub-threshold ($x<x_\text{th}$) or supra-threshold ($x \geq x_\text{th}$ ). See the main text for each firing pattern's magnetic gain parameter values.}\label{fig:bursting}
\end{figure}

Bifurcation diagrams offer a comprehensive method for examining a system's dynamics. It enables us to simultaneously compare periodic and chaotic behaviors across a spectrum of parameter values. In addition to the above classification of firing patterns, the 5D HR neuron model can display regular (periodic) and irregular (chaotic) spiking and bursting activity.

To gain more insights into the dynamics and to further classify the firing patterns, several bifurcation diagrams have been computed, and some of these diagrams (shown in Fig. \ref{fig:master_lya_a}) have been chosen to illustrate the general structures of the system.
The bifurcation diagrams display the inter-spike intervals (ISIs), which represent the time between two consecutive peaks in the spike train of the membrane potential variable 
$x(t)$, plotted against a bifurcation parameter. These bifurcation diagrams reveal that the model exhibits a plethora of dynamical behaviors, including periodic, quasi-periodic, and periodic-doubling leading chaotic dynamics.

The accuracy of the bifurcation diagrams shown in Fig. \ref{fig:master_lya_a} is verified using the maximum Lyapunov exponent. Thus, for each value of the bifurcation parameters, we  computed the corresponding maximum Lyaponuv exponent (LE) defined as \cite{yamakou2020chaotic}:
\begin{align}
\Lambda_\text{max} = \lim_{t \to \infty} \frac{1}{t} \ln \big(|{L(t)}|\big),
\end{align}
where ${L(t)}:=(\delta x^2 + \delta y^2 + \delta z^2 + \delta w^2 + \delta \phi^2)^{1/2}$ is obtained by simultaneously (and numerically) solving system Eq. \eqref{eq:master} and its corresponding variational system given by:
\begin{eqnarray}\label{eq:master_var}
\left\{\begin{array}{lcl}
 \dot{\delta x} &=& (2ax - 3bx^2 - k_{1}\alpha -3 k_1\beta \phi^2)\delta x + \delta y - \delta z    \\
 &-& 6k_1 \beta x \phi \delta \phi,\\
 \dot{\delta y} &=&  -2dx\delta x - \delta y - \sigma \delta w,\\
 \dot{\delta z} &=&  s\theta \delta x - \theta \delta z, \\
 \dot{\delta w} &=&  \mu \gamma \delta y -\rho \mu \delta w, \\
 \dot{\delta\phi} &=&  \delta x - k_{2}\delta \phi, \\
\end{array}\right.
\end{eqnarray}
where the dots on the L.H.S (and in the rest of the paper) represent the differential operator $d/dt$.

The maximum LE $\Lambda_\text{max}$ provides us not only with qualitative insights into the system's behavior but also with a quantitative measure of its stability. A negative max LE ($\Lambda_\text{max}  < 0$) indicates dissipative systems, which exhibit asymptotic stability—the more negative the exponent, the greater the stability. In this context, super-stable fixed points and super-stable periodic orbits have a LE of $\Lambda_\text{max}  = -\infty$. When $\Lambda_\text{max}  = 0$, the system is marginally stable or quasi-periodic, indicating a conservative system. Finally, if $\Lambda_\text{max} > 0$, the orbit is unstable and chaotic, meaning that nearby trajectories will diverge, making the system's evolution highly sensitive to even infinitesimal changes in the initial conditions.

Alongside the bifurcation diagrams in Fig. \ref{fig:master_lya_a}, we present the corresponding variations of the three largest LEs out of five. Each bifurcation diagram is fully traced by the largest LE, $\Lambda_\text{max}$ (shown in blue).  For example, the bifurcation diagram and the corresponding LE spectrum in Fig. \ref{fig:master_lya_a}\textbf{(a)} illustrate periodic firing for $ s \in [3.0, 3.75) $, followed by a period-doubling bifurcation at $ s = 3.85 $ leading to chaos for $ s \in (3.75, 3.85] $, a narrow window of period-3 firing transitioning to chaotic firing via another period-doubling bifurcation, a second instance of period-3 firing leading to chaos in $ s \in [4.7, 4.85] $, and finally, period-2 firing in $ s \in [4.85, 5.0] $.

From Figs. \ref{fig:bursting}, \ref{fig:master_lya_a}\textbf{(c)} and \textbf{(d)}, we can see that changes in magnetic gain parameters $k_1$ and $k_2$ can induce different firing patterns, including periodic, quasi-periodic, and chaotic dynamics. We utilize two-parameter phase diagrams in Fig. \ref{fig:chaos_master} to gain further insight into this dynamical behavior. This colorful diagram is produced by numerically computing  $\Lambda_\text{max}$ on a grid of $200 \times 500$ parameter values.

In Fig. \ref{fig:chaos_master}\textbf{(a)}, the colors represent values of the bifurcation parameters $k_1$ and $k_2$ corresponding to different firing patterns, including supra-threshold bursting (red), sub-threshold bursting (green), supra-threshold spiking (orange), sub-threshold spiking (blue),  and no firing activity (black). In Fig. \ref{fig:chaos_master}\textbf{(b)}, the colors represent the magnitude of $\Lambda_\text{max}$: pink --- stable periodic firing; white --- marginally stable periodic and quasi-periodic firing; and green --- chaotic firing. 
Therefore, from Fig. \ref{fig:chaos_master}\textbf{(a)} and \textbf{(b)}, the values of the bifurcation parameters $k_1$ and $k_2$ that correspond to specific periodic, quasi-periodic, or chaotic firing patterns in the model can be identified.

\begin{figure}[t]
    \centering
        \includegraphics[width=5.0cm,height=7.0cm]{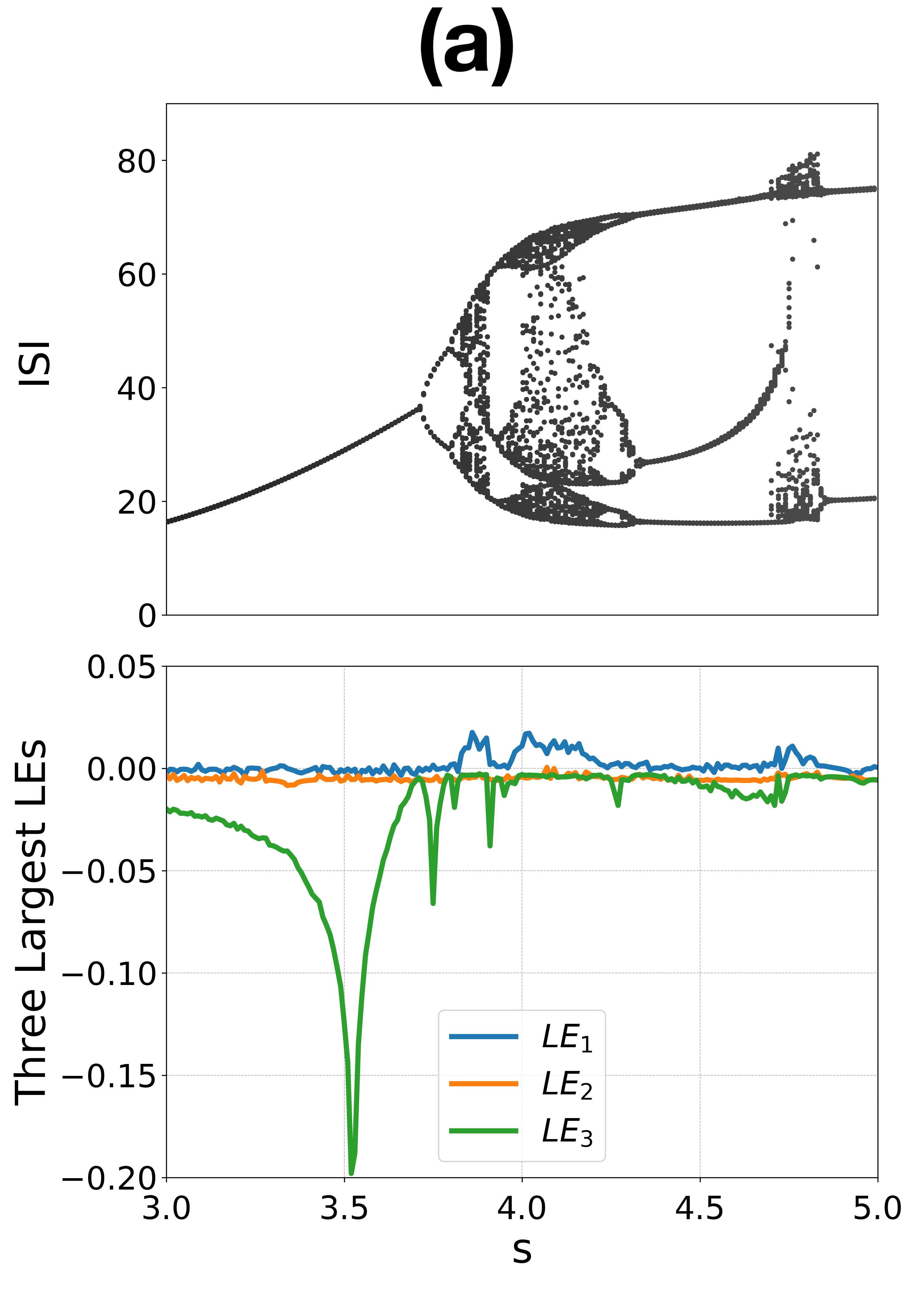}
        \includegraphics[width=5.0cm,height=7.0cm]{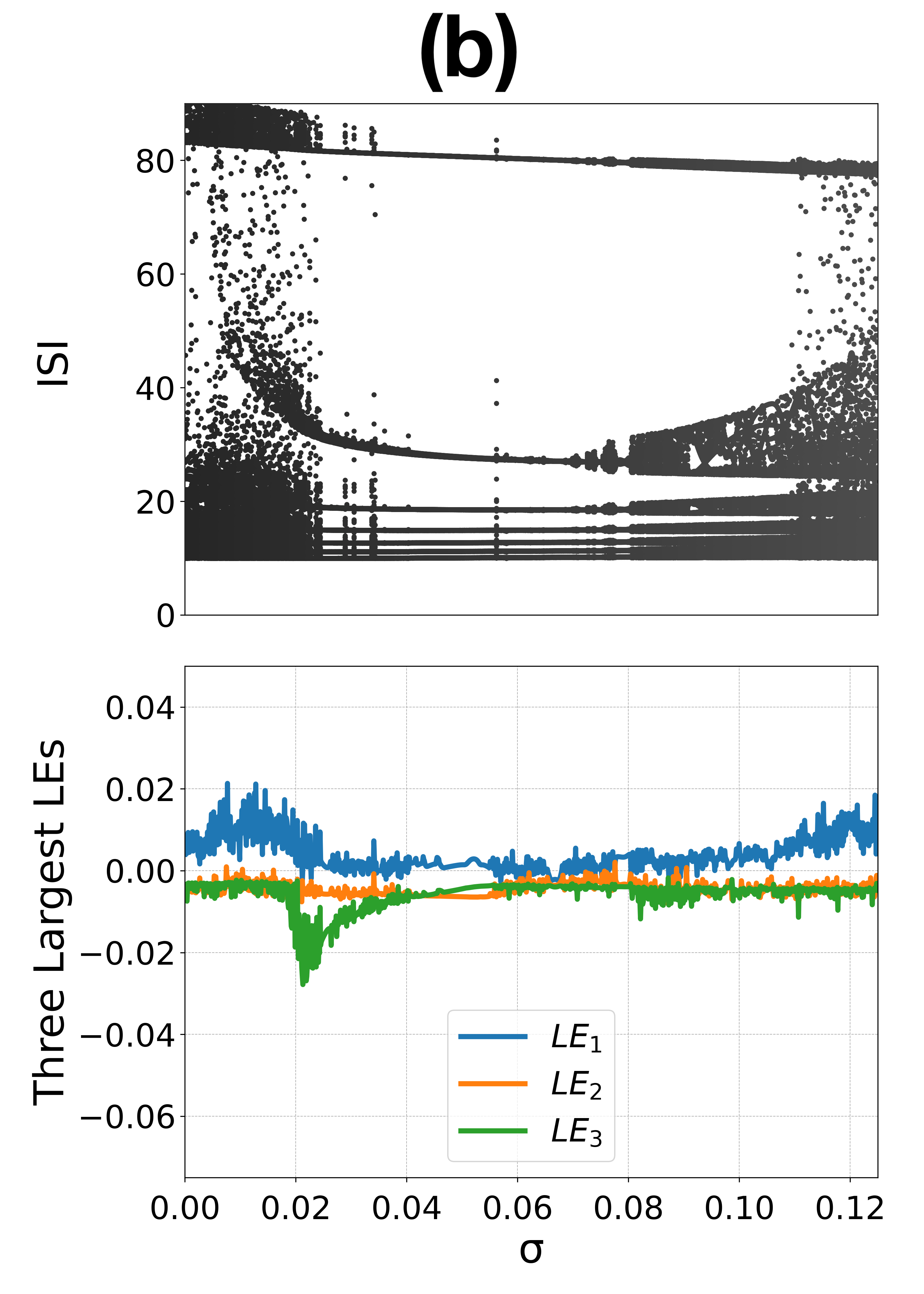}
        \includegraphics[width=5.0cm,height=7.0cm]{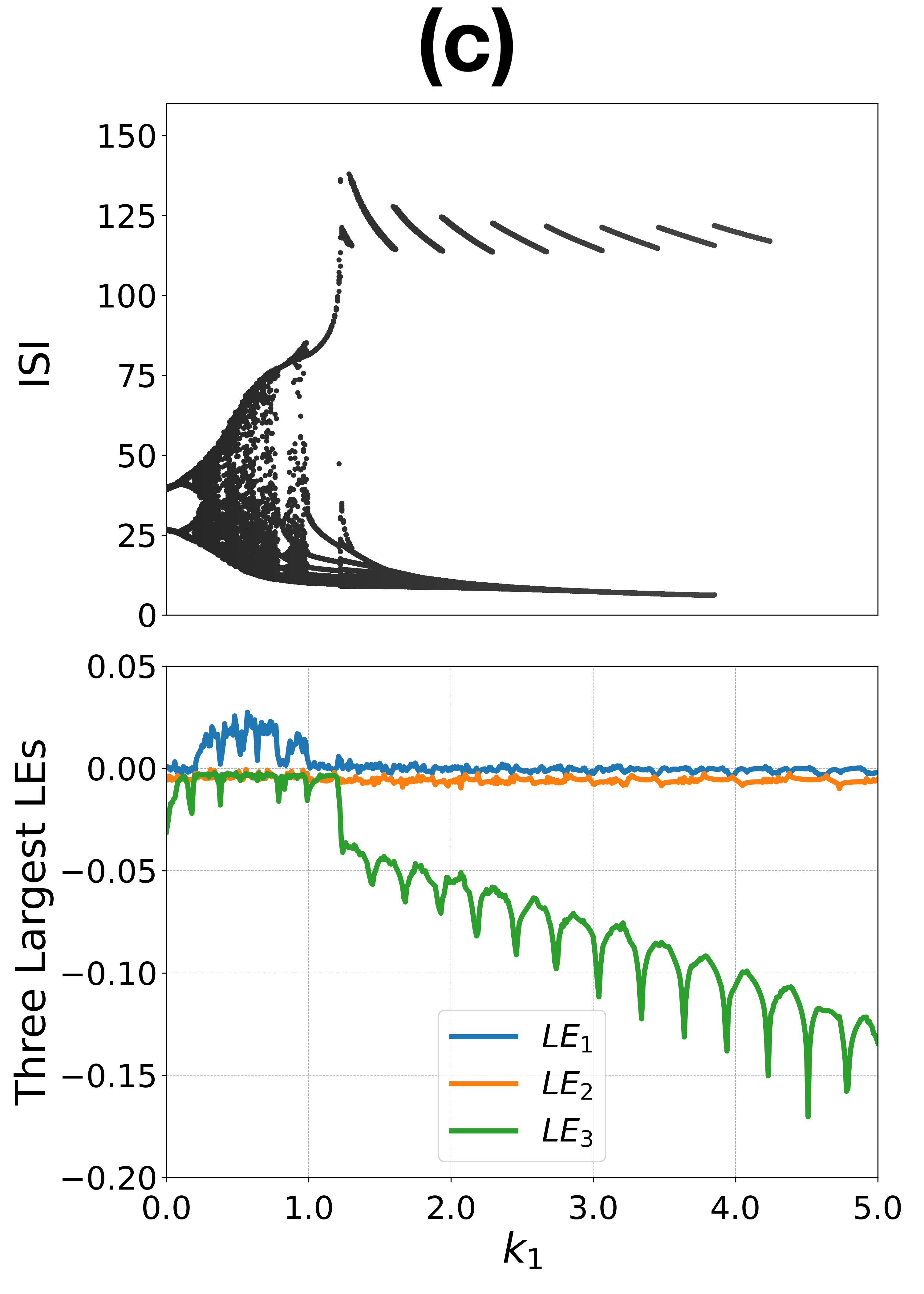}
         \includegraphics[width=5.0cm,height=7.0cm]{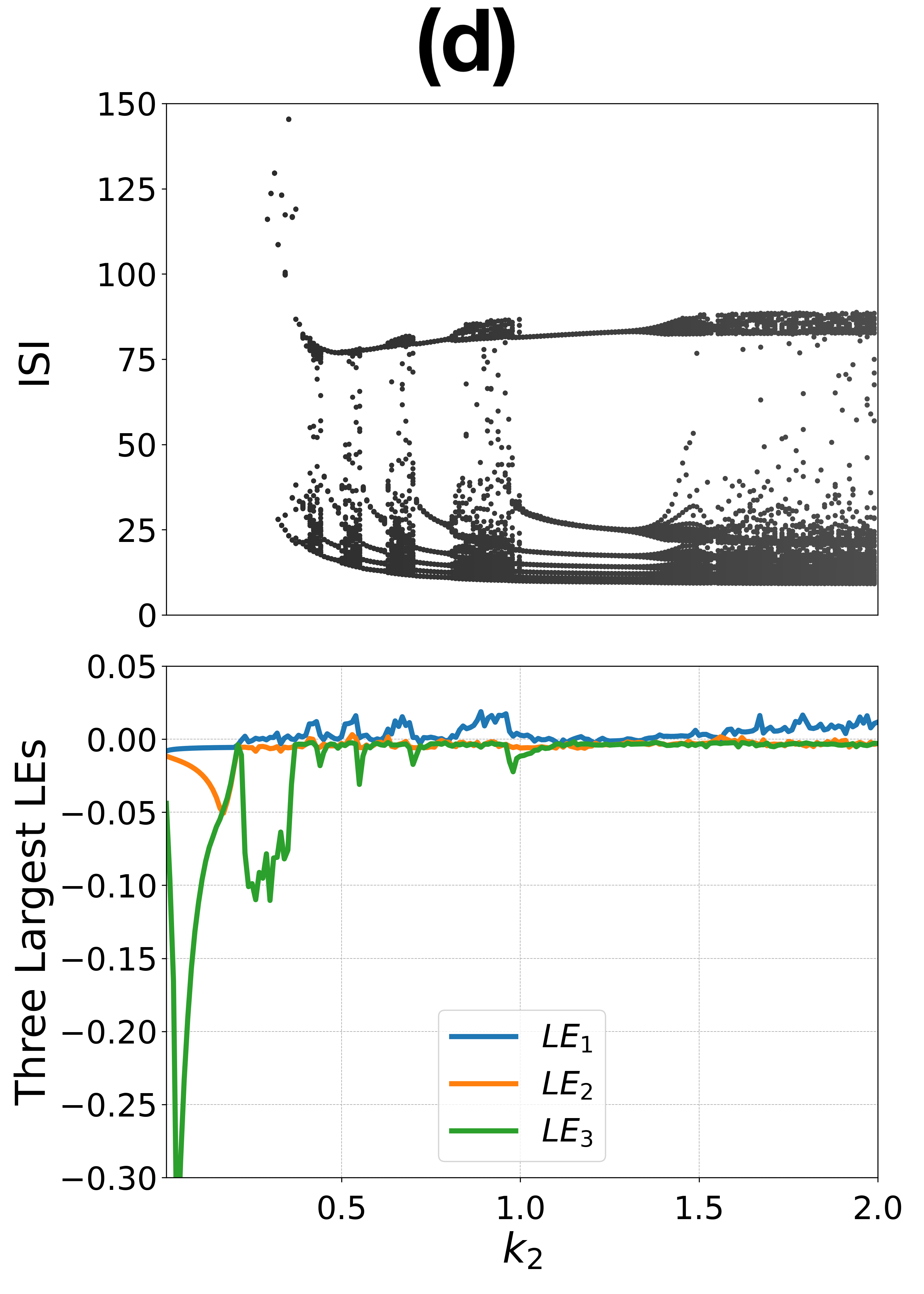}
    \caption{Bifurcation diagrams of the 5D HR neuron model and its three largest Lyapunov exponents as a function of \textbf{(a)} $s$ with $\sigma=0.0278$, $k_1 = 0.7$, $k_2=0.5$; \textbf{(b)} $\sigma$ with $s = 3.875$, $k_1=1.0$, $k_2=1.0$; \textbf{(c)} $k_1$ with $s = 3.875$, $\sigma=0.0278$, $k_2=1.0$; and \textbf{(d)} $k_2$ with $s$ = 3.875, $\sigma=0.0278$, $k_1=1.0$.}\label{fig:master_lya_a}
\end{figure}

\begin{figure}[t]
  \centering
  \includegraphics[width=10.0cm,height=10.0cm]{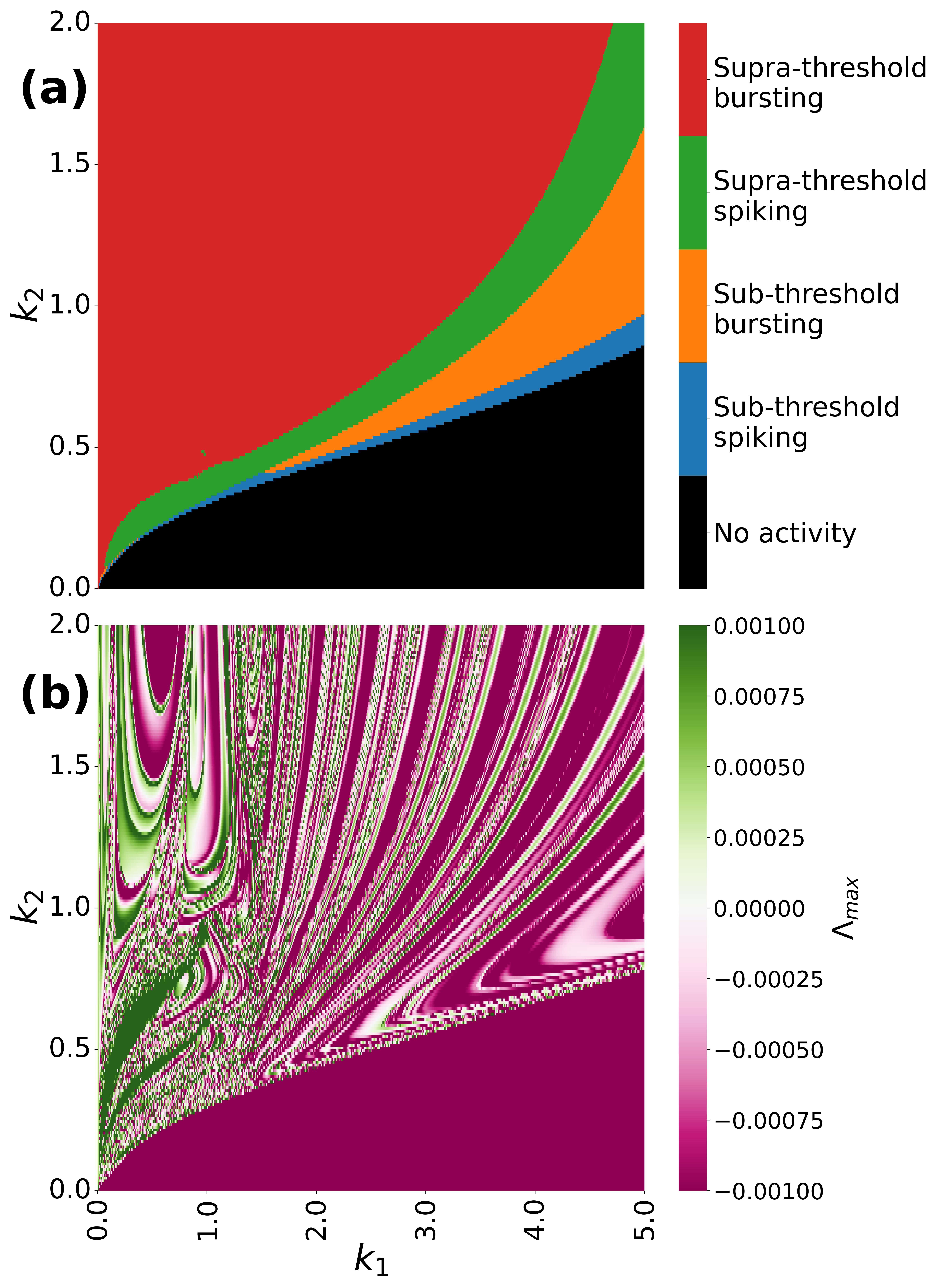}
  \caption{Firing pattern classification of the 5D HR neuron model as a function of the magnetic gain parameters $k_1$ and $k_2$ in panel \textbf{(a)} and the corresponding Lyapunov diagram for the largest Lyapunov exponent $\Lambda_\text{max}$ in panel \textbf{(b)}.}
  \label{fig:chaos_master}
\end{figure}

\section{Dynamical system approach to the reduced-order synchronization problem}\label{sec:DS}
In this section, we define and address the problem of reduced-order synchronization through the lens of dynamical systems theory.

\subsection{Definitions and statement of the problem}
\begin{definition}\label{def_1}:
Two chaotic systems $ x(t) \in \mathbb{R}^n$ and $ y(t) \in \mathbb{R}^n $ are said to have achieved \textit{same order complete synchronization} if there exists a set of initial conditions $ x(t_0) = x_0 $ and $ y(t_0) = y_0 $ such that the trajectories of the systems satisfy
$\lim\limits_{t \to \infty} \| x(t) - y(t) \| = 0$.
\end{definition}
The problem of same-order complete synchronization has been the main focus of previous research \cite{yamakou2020chaotic,arenas2008synchronization,tang2014synchronization,boccaletti2006complex,pecora2015synchronization}, and due to the systems having the same dimension and, hence, structurally comparable,  complete synchronization (as described in Definition \ref{def_1}) becomes relatively easier to achieve.

However, in real-world systems, interacting systems may not always have the same order; i.e., they may have different state space dimensions, making synchronization much harder to achieve than in dynamical systems of the same order. 
This necessitates the definition of a \textit{reduced-order complete synchronization problem}. 
\begin{definition} \label{def_2}: Two chaotic systems $x(t) \in \mathbb{R}^n$ and $y(t) \in \mathbb{R}^m, $ ($n>m$), are said to have achieved \textit{reduced-order complete synchronization} if there exists a set of initial conditions $ x(t_0) = x_0 $ and $ y(t_0) = y_0 $ such that the trajectories of the systems satisfy $\lim\limits_{t \to \infty} \| z(t) - y(t) \| = 0$, where  
$ z(t) \subset x(t)$  and $z(t) \in \mathbb{R}^m$.
\end{definition}

Furthermore, in real-world systems, achieving complete (ideal) synchronization---where the trajectories of interacting systems eventually become exactly equal (i.e., $\lim\limits_{t \to \infty} \| z(t) - y(t) \| = 0$)---can be difficult or even impossible. This difficulty can arise due to factors such as parameter mismatches/uncertainties, noise, different initial conditions and/or dimensions, differing governing equations of the interacting systems, and numerical integration errors during computation. In such scenarios, one often refers to achieving synchronization in a \textit{practical sense} (or $\varepsilon$-synchronization), where the objective is to attain a sufficient level of coordination tailored to a specific application, ensuring the effective operation of systems even if perfect alignment is not feasible. This concept is pivotal in disciplines such as control theory \cite{li2015practical}, secure communications \cite{sheikhan2013synchronization}, and network synchronization \cite{montenbruck2013practical}, where exact alignment may be impracticable or unnecessary, yet effective coordination remains paramount.

\begin{definition}\label{def_3}: 
Two chaotic systems $x(t) \in \mathbb{R}^n$ and $y(t) \in \mathbb{R}^m, $ ($n>m$), are said to have achieved \textit{reduced-order practical (or $\varepsilon-$) synchronization} if there exists a set of initial conditions $ x(t_0) = x_0 $ and $ y(t_0) = y_0 $ and $0 \leq \varepsilon \ll 1$ such that the trajectories of the systems satisfy $\lim\limits_{t \to \infty} \| z(t) - y(t) \| \leq \varepsilon$, where  $ z(t) \subset x(t)$  and $z(t) \in \mathbb{R}^m$.
\end{definition}

In this paper, we address the problem of \textit{reduced-order synchronization} and in particular, we are interested in the reduced-order synchronization problem between two chaotic and parameter-mismatched HR neuron models: the 5D HR neuron model defined by Eq. \eqref{eq:master} (referred to as the drive system) \cite{yamakou2020chaotic} and a 4D HR neuron model (referred to as the response system) given by \cite{lv2016model}:
\begin{eqnarray}
\left\{\begin{array}{lcl}\label{eq:response}
\dot{x_r} &=& \overline{a}x_r^2 - \overline{b}x_r^3 + y_r - z_r - k_{1}(\alpha+3\beta \phi_r^2)x_r\\
&+& I + U_x(t),\\
\dot{y_r} &=&   c - \overline{d}x_r^2 - y_r  + U_y(t), \\[2.0mm]
\dot{z_r}&=&   \overline{\theta}[s(x_r - x_0) - z_r] + U_z(t),\\[2.0mm]
\dot{\phi_r} &=&  x_r - k_{2}\phi_r + U_{\phi}(t),
\end{array}\right.
\end{eqnarray}
where the subscripts $ r $ distinguish variables of the response system from those of the drive system. Furthermore, $ \overline{a} $, $ \overline{b} $, $ \overline{d} $, and $\overline{\theta} $ represent the mismatched parameters of the response system that need to be estimated. The vector $ \mathbf{U}(t) = [U_x(t), U_y(t), U_z(t), U_{\phi}(t)]^T $ denotes coupling forces, which are active feedback control functions to be designed to achieve reduced-order synchronization of the chaotic systems described by Eqs. \eqref{eq:master} and \eqref{eq:response}. The control objective is to synchronize all states of the response system with the corresponding states of the drive system's canonical projection, thereby achieving reduced-order synchronization.

It is worth pointing out that when (i) the corresponding parameters of the drive and response systems are fixed at the same values, (ii) the additional parameters of the drive system are fixed to the values stated earlier (or see caption of Fig. \ref{fig:portraits} below), and (iii) the drive and response systems are set at same arbitrary initial conditions, i.e., $x(0)=x_r(0)=0.1$, $y(0)=y_r(0)=0.2$, $z(0)=z_r(0)=0.3$, $w(0)=0.1$, and $\phi(0)=\phi_r(0)=0.2$, the two systems exhibit different attractors. 
Figures \ref{fig:portraits}\textbf{(a)}-\textbf{(b)} show the phase portraits (attractors) of the 5D drive system (Eq. \eqref{eq:master}) and figures \ref{fig:portraits}\textbf{(c)}-\textbf{(d)} show the phase portraits of the 4D response system (Eq. \eqref{eq:response}). In the response system, $ \mathbf{U}(t) = [U_x(t), U_y(t), U_z(t), U_{\phi}(t)]^T  = (0, 0, 0, 0) $ for all $t \geq 0$, indicating the absence of controllers. As can be observed, the attractors of the drive and response systems differ significantly. Thus, there is no complete synchronization.

\begin{figure}[t]
    \centering
        \includegraphics[width=5.0cm,height=5.0cm]{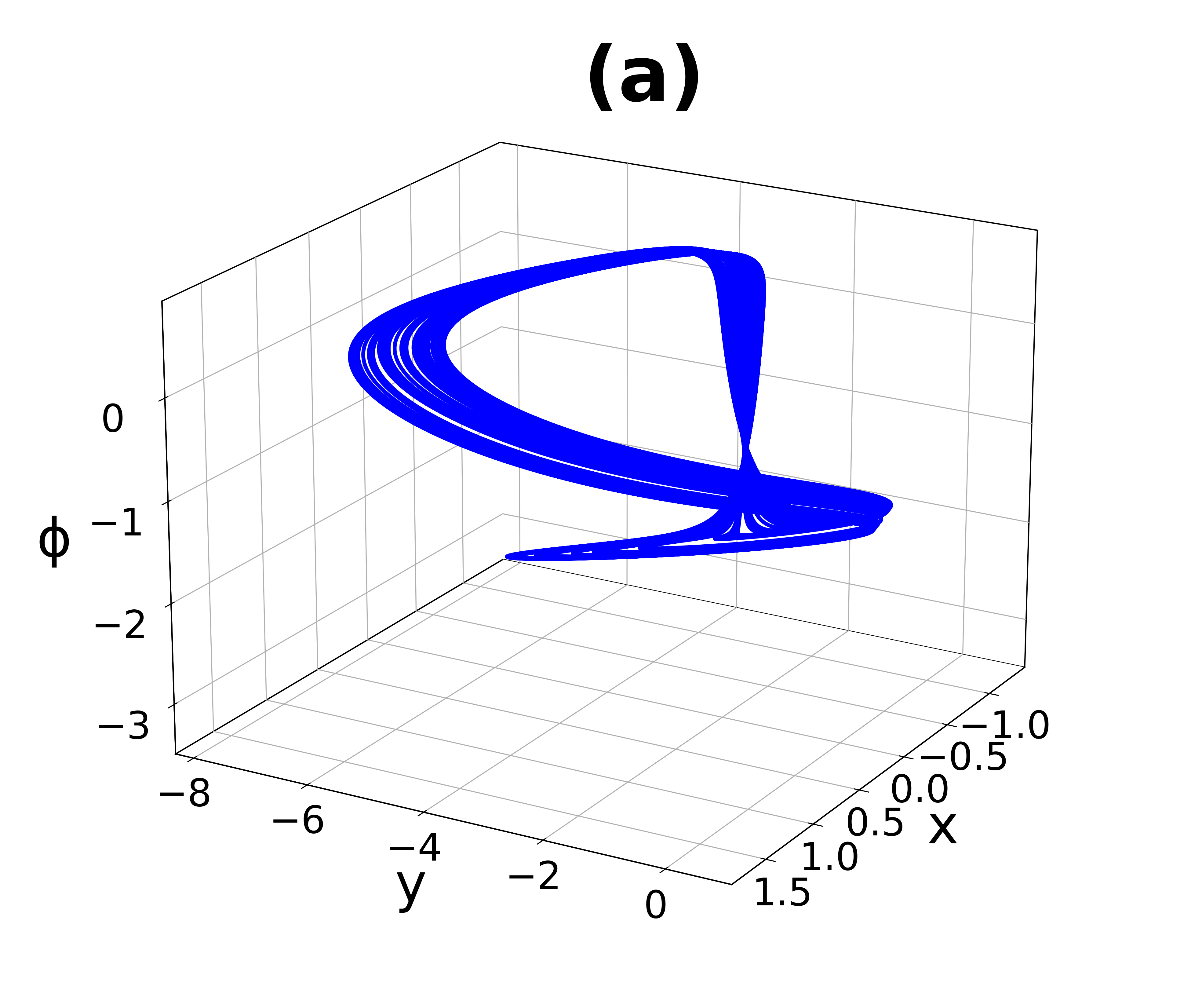}
        \includegraphics[width=5.0cm,height=5.0cm]{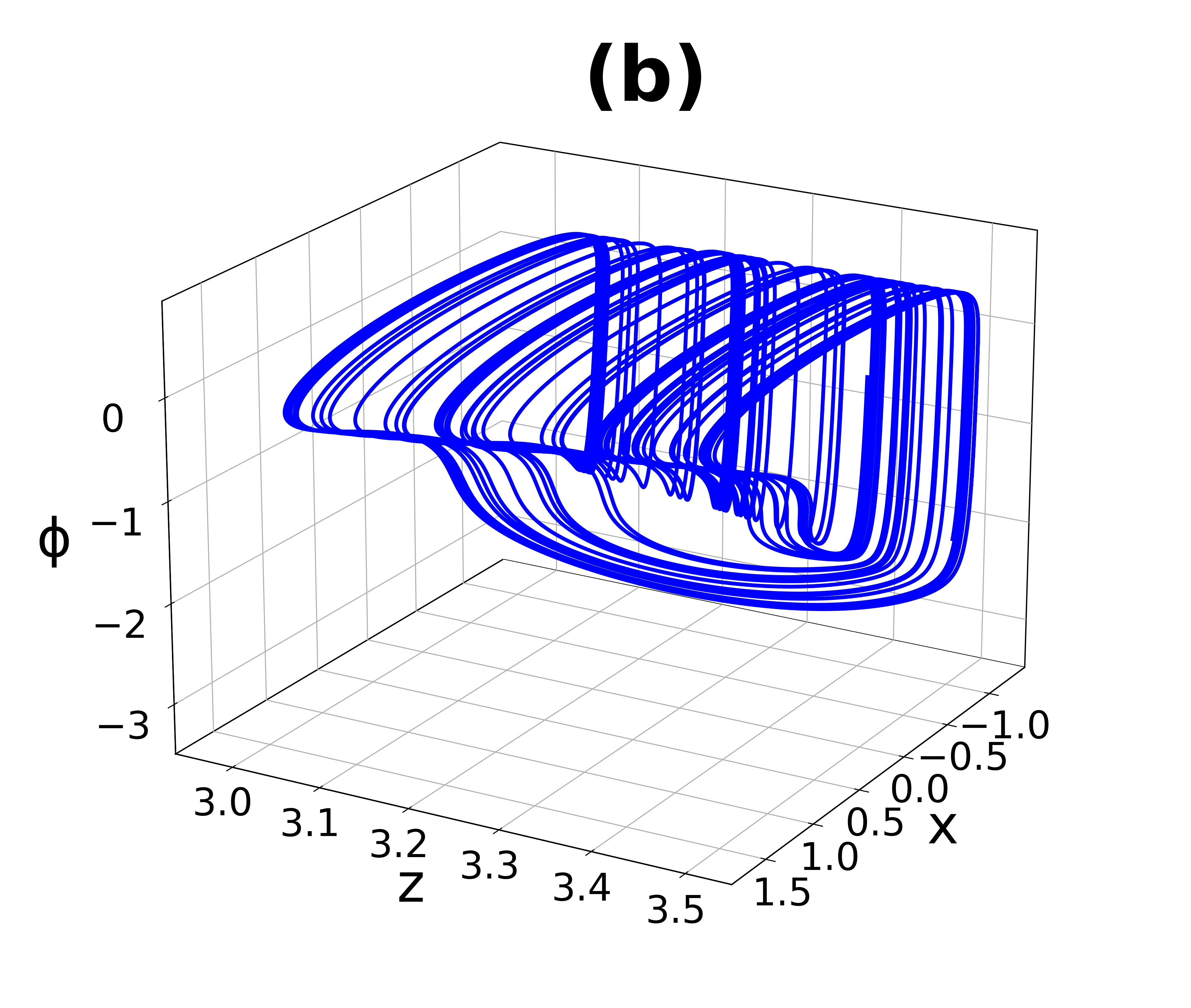}        \includegraphics[width=5.0cm,height=5.0cm]{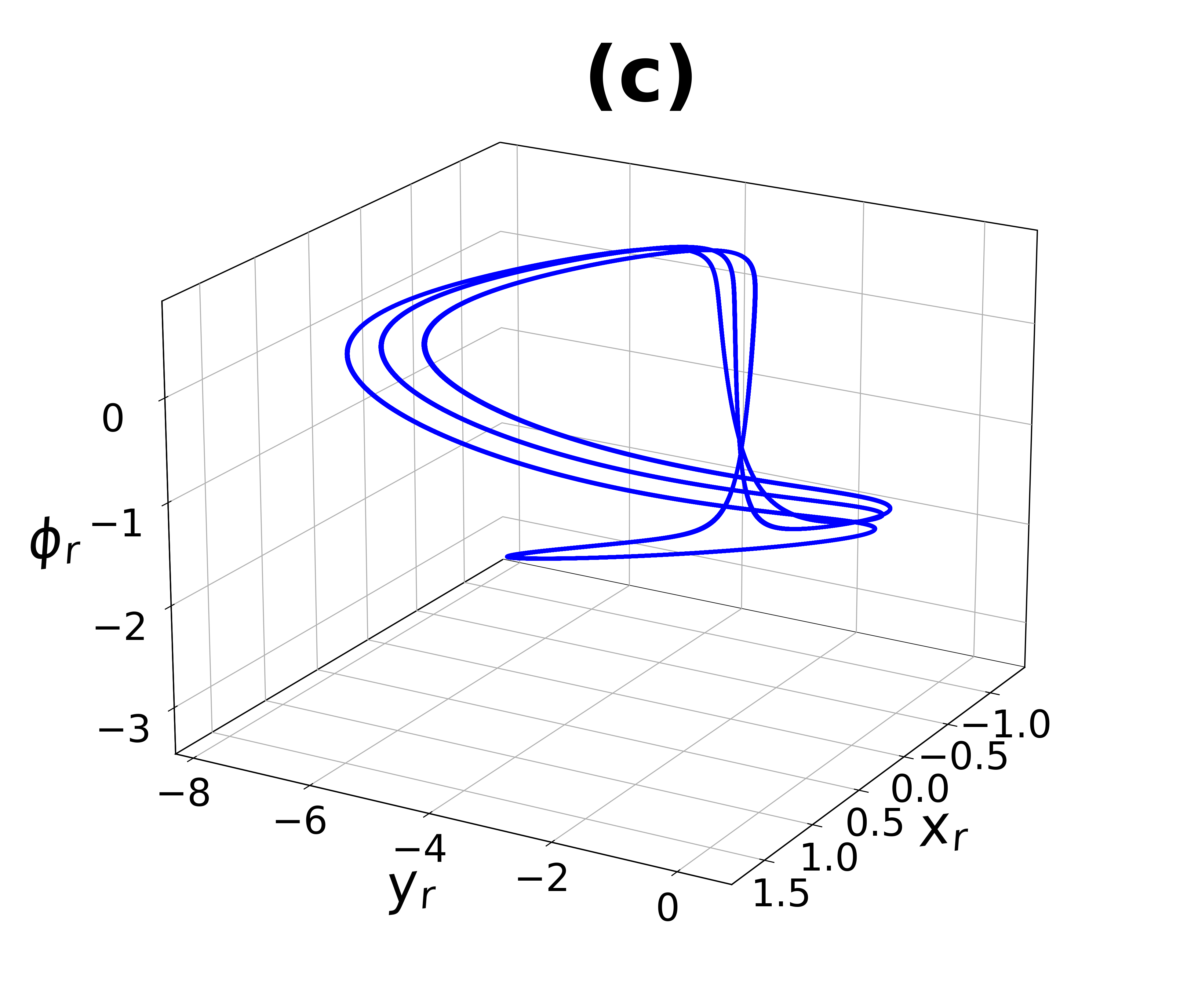}
        \includegraphics[width=5.0cm,height=5.0cm]{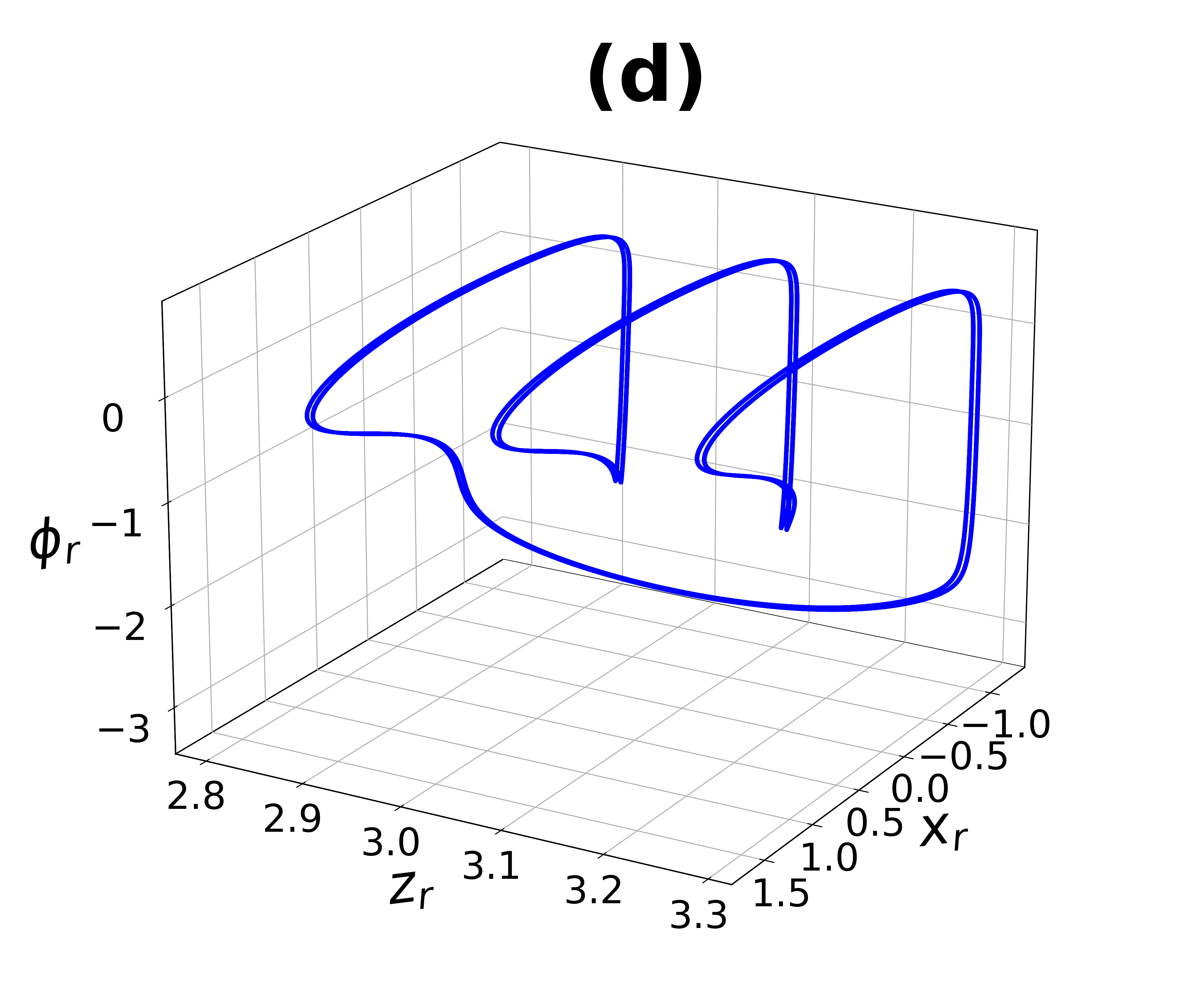}  
    \caption{Bounded chaotic attractors  of the drive system projected on the \textbf{(a)} $\phi yx$-space \textbf{(b)} and $\phi zx$-space. Also the response system projected on the \textbf{(c)}  $\phi_r y_rx_r$-space and \textbf{(d)} $\phi_r z_rx_r$-space. Parameters: $a = 3.0$, $b = 1.0$, $\alpha = 0.1$, $\beta = 0.02$, $c = 1.0$, $d = 5.0$, $\sigma = 0.0278$, $\theta = 0.006$,  $x_0 = -1.56$, $y_0 = -1.619$, $\mu = 0.0009$, $\gamma = 3.0$, $\rho = 0.9573$, $I = 3.1$, $s=4.75$, $k_1 = 0.08$ and $k_2=0.4$.}\label{fig:portraits}
\end{figure}

The reduced-order synchronization of the response system Eq. \eqref{eq:response} and the drive system  Eq. \eqref{eq:master} occurs when the corresponding variables of the two systems asymptotically exhibit identical behavior, i.e.,
\begin{equation} \label{errors_0}
\begin{split}
\begin{cases}
\lim\limits_{t \to \infty} \| x_r(t) - x(t) \| = 0,\\
\lim\limits_{t \to \infty} \| y_r(t) - y(t) \| = 0,\\
\lim\limits_{t \to \infty} \| z_r(t) - z(t) \| = 0,\\
\lim\limits_{t \to \infty} \| \phi_r(t) - \phi(t) \| = 0,
\end{cases}
\end{split}
\end{equation}
 for initial conditions chosen in some neighborhood
of the synchronization manifold, denoted by $\mathcal{M}_s$. Thus, by introducing coordinates transformation to $\mathcal{M}_s$, defined by the error vector $\mathbf{e}(t)=[e_x(t),e_y(t),e_z(t),e_{\phi}(t)]^T$, between the response and drive system as
\begin{equation} \label{errors}
\begin{split}
\begin{cases}
e_x(t) = x_r(t)  - x(t), \\
e_y(t) = y_r(t)  - y(t), \\
e_z(t) = z_r(t) - z(t), \\
e_{\phi}(t) = \phi_r(t)  - \phi(t), \\
\end{cases}
\end{split}
\end{equation}
we obtain the dynamics of the most significant transverse displacements from $\mathcal{M}_s$ as follows:
\begin{eqnarray}
\left\{\begin{array}{lcl}\label{eq:errordynamic_full}
\dot{e_x} &=& (2ax_r - 3bx_r^2 -k_{1}\alpha - 3k_1 \beta \phi_r^2)e_x + e_y - e_z \\
&-& 6k_1 \beta x_r \phi_re_{\phi} + [\overline{a}-a]x_r^2   - [\overline{b}-b]x_r^3 + U_x \\ 
\dot{e_y} &=& - 2dx_re_x - e_y -[\overline{d}-d]x_r^2 + \sigma w + U_y,\\
\dot{e_z} &=& \theta se_x - \theta e_z + [\overline{\theta}-\theta][s(x_r-x_0) - z_s] + U_z,\\
\dot{e_{\phi}} &=& e_x  - k_{2}e_{\phi} + U_{\phi}.\\
\end{array}\right.
\end{eqnarray}

Thus, the reduced-order synchronization problem between the drive and response systems given, respectively, by Eqs. \eqref{eq:master} and \eqref{eq:response}, reduces to finding an adaptive control command $\mathbf{U}=(U_x,U_y,U_z,U_{\phi})^T$ that ensures the error system in Eq. \eqref{eq:errordynamic_full} is asymptotically stable at the origin for any set initial condition $\mathbf{e}(0)=[e_x(0),e_y(0),e_z(0),e_{\phi}(0)]^T$. If the adaptive control command $\mathbf{U}=(U_x,U_y,U_z,U_{\phi})^T$ successfully stabilizes the error system in Eq. \eqref{eq:errordynamic_full} at the origin, that is at $\mathbf{e}(t)=[e_x(t),e_y(t),e_z(t),e_{\phi}(t)]^T = (0,0,0,0)^T$ for all $ t \ge t_0 \ge 0 $, where $ t_0 $  is the time when the control is activated, then reduced-order synchronization defined in Eq. \eqref{errors_0} would be achieved. 

\begin{theorem}
The 4D response system in Eq. \eqref{eq:response} and the 5D drive system of Eq. \eqref{eq:master} achieve reduced-order synchronization  following the adaptive controllers
\begin{equation}\label{t1:1}
U_i = -g_ie_i, \:\:\: \dot{g_i}(t)=e^2_i, \:\:\: i = x, y,z,\phi,
\end{equation}
and the update laws of the estimated parameters
\begin{equation}
\begin{split}
\begin{cases}\label{t1:2}
\dot{\overline{a}}=-x_r^2e_x, \\
\dot{\overline{b}}=x_r^3e_x, \\
\dot{\overline{d}}=x_r^2e_y,  \\
\dot{\overline{\theta}}=-[s(x_r - x_0)-z_r]e_z, \\
\end{cases}
\end{split}
\end{equation}
where $g_i$ is the feedback gain which will be estimated.
\end{theorem}

\begin{proof}
Under the adaptive control laws in Eq. \eqref{t1:1} and the update laws of the mismatched parameters in Eq. \eqref{t1:2}, the error
dynamics in Eq. \eqref{eq:errordynamic_full} reads
\begin{eqnarray}
\left\{\begin{array}{lcl}\label{eq:errordynamic_update}
\dot{e_x} &=& (2ax_r - 3bx_r^2 -k_{1}\alpha - 3k_1 \beta \phi_r^2)e_x + e_y - e_z \\
&-& 6k_1 \beta x_r \phi_re_{\phi} + [\overline{a}-a]x_r^2   - [\overline{b}-b]x_r^3 -g_xe_x,\\ 
\dot{e_y} &=& - 2dx_re_x - e_y -[\overline{d}-d]x_r^2 + \sigma w -g_ye_y,\\
\dot{e_z} &=& \theta se_x - \theta e_z + [\overline{\theta}-\theta][s(x_r-x_0) - z_s] -g_ze_z,\\
\dot{e_{\phi}} &=& e_x  - k_{2}e_{\phi} -g_{\phi}e_{\phi},\\
\dot{g_x} &=& e_x^2,\\
\dot{g_y} &=& e_y^2,\\
\dot{g_z} &=& e_z^2,\\
\dot{g_{\phi}} &=& e_{\phi}^2,\\
\dot{\overline{a}}&=&-x_r^2e_x, \\
\dot{\overline{b}}&=&x_r^3e_x, \\
\dot{\overline{d}}&=&x_r^2e_y,  \\
\dot{\overline{\theta}}&=&-[s(x_r - x_0)-z_r]e_z. 
\end{array}\right.
\end{eqnarray}
We construct a continuous, positive definite Lyapunov function $V$ of the form
\begin{eqnarray}\label{lya_fun}
\begin{split}
V&= \dfrac{1}{2}\bigg[e_x^2 +e_y^2+e_z^2 + e_{\phi}^2 + (g_x - l_x)^2 + (g_y - l_y)^2\\
&+ (g_z - l_z)^2 + (g_{\phi} - l_{\phi})^2 + (\overline{a}-a)^2  \\
&+ (\overline{b}-b)^2 + (\overline{d}-d)^2 + (\overline{\theta}-\theta)^2\bigg]\geq0, \\
\end{split}\label{eq_lya_func}
\end{eqnarray}
where $l_i$ ($i=x,z,y,\phi$) are all positive constants. The time derivative of $V$ (denoted $\dot{V}$) along the trajectories of the error dynamical system in Eq. \eqref{eq:errordynamic_update} is given by
\begin{eqnarray}\label{eq:32}
\begin{split}
\dot{V} &= \big(2ax_r - 3bx_r^2 - k_1\alpha - 3k_1\beta \phi_r^2  - l_x\big)e_x^2 \\
& - \big(1+l_y\big)e_y^2 - \big(\theta + l_z\big)e_z^2 - \big(k_2 + l_{\phi}\big)e_{\phi}^2\\
& + \big(1-2dx_r\big)e_xe_y + \big(s\theta -1\big)e_xe_z \\
& + \big(1-6k_1\beta x_r \phi_s\big)e_xe_{\phi} + \sigma w e_y.
\end{split}
\end{eqnarray}
Since the drive system and the response system \cite{lv2016model} are chaotic, there exist positive real constants $J>0$ and $J_r>0$, such that the state trajectories of systems are globally bounded, i.e., $|i(t)| \leq J$ and $|i_r(t)| \leq J_r$, $i=x,y,z,\phi$, hold for all $t \geq 0$, respectively. 

It is worth noting that the $w$ variable of the drive system is always negative, i.e., $w(t)<0$ for all $t\geq0$. Even with a positive initial condition, i.e., $w(0)>0$, we have $w(t)<0$ for all $t>t_0$, where $t_0$ is the time-lapse of transient solutions (see the time series of the $w(t)$ variable in Fig. \ref{fig:bursting}). Therefore, one can impose the bound $w(t) \le - J <0$ (since $J>0$) for all $t>t_0$.
Applying these bounds term-wise in Eq. \eqref{eq:32}, we obtain
\begin{eqnarray}\label{eq:lya_dev}
\begin{split}
\dot{V} & \leq \big(2|a|J_r + 3|b|J_r^2 - |k_1\alpha| + 3|k_1\beta|J_r^2  - |l_x|\big)e_x^2 \\
& - |1 + l_y|e_y^2 - |\theta + l_z|e_z^2 - |k_2+l_{\phi}|e^2_{\phi}  \\
 &+ \big(1+2|d|J_r\big)|e_x||e_y|  + |s\theta -1||e_x||e_z| \\
 &+ \big(1+6|k_1\beta|J_r^2\big)|e_x||e_{\phi}| - J|\sigma||e_y|,
\end{split}
\end{eqnarray}
which is compactly written as 
\begin{eqnarray}\label{eq:lyapunov_matrix_equal}
\begin{split}
\dot{V}  \leq  -\mathbf{e}^T\mathbf{K}\mathbf{e}- J|\sigma||e_y|,
\end{split}
\end{eqnarray}
where $\mathbf{e}^T = (e_x, e_y, e_z, e_{\phi})$ and 
\begin{equation}
\mathbf{K}=\begin{bmatrix} 
A_{xx} & \:\:\:\:\: A_{xy} & \:\:\:\:\:A_{xz} & \:\:\:\:\: A_{x\phi}\\
A_{yx} & \:\:\:\:\:|1 + l_y|      & \:\:\:\:\:0       & \:\:\:\:\:0\\
A_{zx} &\:\:\:\:\:0       & \:\:\:\:\:|\theta + l_z| & \:\:\:\:\:0\\
A_{\phi x} & \:\:\:\:\:0       &\:\:\:\:\:0       & \:\:\:\:\:|k_2+l_{\phi}|\\
\end{bmatrix},
\end{equation}
is a symmetric matrix with 
\begin{eqnarray}
\left\{\begin{array}{lcl}
A_{xx} &=& -2|a|J_r - 3|b|J_r^2 + |k_1\alpha| - 3|k_1\beta|J_r^2  + |l_x|,\\[2pt]
A_{yx} &=& A_{xy} =  -\frac{1}{2}\big(1+2|d|J_r\big),\\[2pt]
A_{zx} &=& A_{xz} = -\frac{1}{2}|s\theta -1|,\\[2pt]
A_{\phi x} &=& A_{x\phi} = -\frac{1}{2}\big(1+6|k_1\beta|J_r^2\big).
\end{array}\right.
\end{eqnarray}

The origin of the error dynamical system in Eq. \eqref{eq:errordynamic_full} is stable in the Lyapunov sense, if $\dot{V}$ in Eq. \eqref{eq:lyapunov_matrix_equal} is negative semi-definite. 
Since $J|\sigma||e_y|$ is positive definite, it suffices to show that the matrix  ${\mathbf{K}}$ is also positive definite to have a negative semi-definite $\dot{V}$.

\begin{lemma}
(Sylvester's criterion). A real-symmetric matrix $\mathbf{A}$ is positive definite if and only if all the leading principal minors of $\mathbf{A}$ are positive. 
\end{lemma}
The leading principal minors of the matrix ${\mathbf{K}}$ are
\begin{eqnarray}
\left\{\begin{array}{lcl}
\mathbf{D_x} &=& -2|a|J_r - 3|b|J_r^2 + |k_1\alpha| - 3|k_1\beta|J_r^2  + |l_x|,\\[2.0mm]
\mathbf{D_y} &=& \mathbf{D_x}|1+l_y| - \frac{1}{4}\big(1+ 2|d|J_r\big)^2,\\[2.0mm]
\mathbf{D_z} &=& \mathbf{D_y}|\theta +l_z| - \frac{1}{4}|1+l_y||s\theta  -1|^2,\\[2.0mm]
\mathbf{D_{\phi}} &=& \mathbf{D_z}|k_2+l_{\phi}| - \frac{1}{4}|1+l_{y}||\theta +l_z|\big(1+ 6 |\beta k_2| J_r^2\big)^2,
\end{array}\right.
\end{eqnarray}
and $\mathbf{D_\text{i}} >0$, $i=x,y,z,\phi$, if and only if

\begin{eqnarray}\label{eq:39}
\left\{\begin{array}{lcl}
|l_x| > 2|a|J_r + 3|b|J_r^2 - |k_1\alpha| + 3|k_1\beta| J_r^2 =: H_1, \\[2.0mm]
|l_x| > H_1  + \dfrac{\big(1+ 2|d|J_r\big)^2}{4|1+ l_y|}   =: H_2, \\[4.0mm]
|l_x| > H_2 + \dfrac{|s\theta -1|^2}{4|\theta +l_z|}  =: H_3,\\[4.0mm]
|l_x| > H_3 + \dfrac{\big(1+6|\beta k_2| J_r^2\big)^2}{4|k_2+ l_{\phi}|}  =: H_4.\\
\end{array}\right.
\end{eqnarray}

From the conditions $l_i>0$, $i=x,y,z,\phi$ (given after Eq. \eqref{lya_fun}) and ordering $H_1<H_2\leq H_3<H_4$  obtained from Eq. \eqref{eq:39}, it is easy to see that
\begin{equation}\label{eq:40}
\mathbf{D_\text{i}}>0, \:\:\:i=x,y,z,\phi\:\:\:\Longleftrightarrow\:\:\:l_x > H_4.
\end{equation}
Thus, there always exist positive constants $l_i$, $i=y,z,\phi$ and suitable values of the parameters $a$, $b$, $c$, $d$, $\alpha$, $\beta$,
$r$, $s$, $\sigma$, $k_1$, and $k_2$ satisfying the condition in Eq.\eqref{eq:40} so that  $\dot{V} \le 0$, and hence $V$ is positive and decrescent. Consequently, the fixed point $e_i=0$, $i=x,z,y,\phi$ of the error dynamical system in Eq. \eqref{eq:errordynamic_update} is stable when $g_i=g_i^*$, $i=x,z,y,\phi$, and $\overline{a}=a$, $\overline{b}=b$, $\overline{d}=d$, $\overline{\theta}=\theta$.

However, the asymptotic stability of this fixed point is not yet guaranteed. To guarantee that transverse perturbations decay to the synchronization manifold without any transient growth, we observe from  Eqs. \eqref{lya_fun} and \eqref{eq:lyapunov_matrix_equal} that $e_i\in L^{\infty}$, and hence $g_i\in L^{\infty}$. Since the trajectories of systems in Eqs. \eqref{eq:master} and \eqref{eq:response} are bounded, it follows that $(\overline{a}-a)$, $(\overline{b}-b)$, $(\overline{d}-d)$, $(\overline{\theta}-\theta)\in L^{\infty}$, where $a$, $b$, $d$, and $\theta$ are the estimated values of the unknown $\overline{a}$, $\overline{b}$, $\overline{d}$, and $\overline{\theta}$, respectively. It also follows that the controllers $U_i\in L^{\infty}$, $i=x,y,z,\phi$. Thus, from the error dynamical system in Eq. \eqref{eq:errordynamic_update}, we have that $\dot{e_i}\in L^{\infty}$, $i=x,y,z,\phi$.

Furthermore, we split $\dot{V}$ in 
Eq. \eqref{eq:lyapunov_matrix_equal} into two parts: $\dot{V} \leq \dot{V}_1 + \dot{V}_2$, where $\dot{V}_1= -\mathbf{e}^T\mathbf{K}\mathbf{e}$ and $\dot{V}_2  = - J|\sigma||e_y|$. Using $\dot{V}_1$, we establish the bound:
\begin{eqnarray}\label{eq:41}
\begin{split}
\begin{aligned}
\int_0^t \mathbf{e}(s)^T \mathbf{K}\mathbf{e}(s) ds &= -\int_0^t \dot{V}_1(s) ds\\&=V_1(0) - V_1(t) \le V_1(0).
\end{aligned}
\end{split}
\end{eqnarray}
\begin{lemma}\label{Rayleigh quotient} (Rayleigh quotient). Let $\mathbf{A}$ be a symmetric square matrix and $\mathbf{x}$ be a non-zero vector, then $\lambda_{min}(\mathbf{A}) \|\mathbf{x}\|^2 \le \mathbf{x}^T\mathbf{A}\mathbf{x}$, where $\lambda_{min}$ is the minimum eigenvalue of $\mathbf{A}$.
\end{lemma}

Combining Eq. \eqref{eq:41} and Lemma \ref{Rayleigh quotient}, we obtain: 
\begin{equation}\label{eq:42}
\begin{split}
\begin{aligned}
\int_0^t \|\textbf{e}(s)\| ds \le V_1(0)/\lambda_{min}(\mathbf{K}) 
\end{aligned}
\end{split}
\end{equation}
where $\lambda_{min}(\mathbf{K})$ is the minimum eigenvalue of the positive definite matrix $\mathbf{K}$. 
\begin{lemma}\label{Barbalat's Lemma}(Barbalat's Lemma). Let $f(t)$  be a uniformly continuous function with values in an interval $[0,\infty)$ and $\lim\limits_{t\to\infty} \int_0^t f(s) ds $ be finite, then $\lim\limits_{t\to\infty} f(t) \rightarrow 0$. 
\end{lemma}
Combining Eqs. \eqref{eq:41} and \eqref{eq:42} and Lemma \ref{Barbalat's Lemma}, we conclude that for any given initial conditions,
\begin{eqnarray}\label{eq:43}
\left\{\begin{array}{lcl}
\lim\limits_{t\rightarrow\infty}\|e_i(t)\| = 0, i=x,y,z,\phi,\\[2.0mm]
\lim\limits_{t\rightarrow\infty} g_i(t) \rightarrow g_i^*, i=x,y,z,\phi,\\[2.0mm]
\lim\limits_{t\rightarrow\infty} \overline{a} \rightarrow  a, \\[2.0mm]
\lim\limits_{t\rightarrow\infty} \overline{b} \rightarrow  b,  \\[2.0mm]
\lim\limits_{t\rightarrow\infty} \overline{d} \rightarrow  d,  \\[2.0mm]
\lim\limits_{t\rightarrow\infty} \overline{\theta} \rightarrow  \theta.  
\end{array}\right.
\end{eqnarray}
Hence, the 4D response system in Eq. \eqref{eq:response} can synchronize in reduced-order the 5D drive system in Eq. \eqref{eq:master} globally and asymptotically under the adaptive control laws in Eq. \eqref{t1:1} and the update laws in Eq. \eqref{t1:2}. This completes the proof. \qed

\end{proof}

\begin{remark}
The convergence of the feedback gains $g_i$ to four
constants $g_i^*$, $i=x,y,z,{\phi}$, respectively, depends only on the
initial condition $g_i(0)$, $i=x,y,z,\phi$. 
\end{remark}
\begin{remark}
The convergence of the three estimated parameters $\overline{a}$, $\overline{b}$, and $\overline{d}$ to the three constants $a$, $b$, and $d$, respectively, also depend only on the initial condition $\overline{a}(0)$, $\overline{b}(0)$, and $\overline{d}(0)$. However, convergence of estimated parameter $\overline{\theta}$ to the constant $\theta$ depends on not only the initial condition $\overline{\theta}(0)$, but also on the values of the parameters $s$ and $x_0$
.\end{remark}

\subsection{Numerical simulations}
Numerical examples are presented in this subsection to verify the effectiveness of the proposed adaptive reduced-order synchronization scheme. In the simulations,  the initial conditions
of the drive system (in Eq. \eqref{eq:master}) and response system (in Eq. \eqref{eq:response}) are set to be  $[x(0),y(0),z(0),\\w(0),\phi(0)]=[1.0,0.5,1.3,-0.5,-1.2]$ and $[x_r(0),y_r(0),\\z_r(0),\phi_r(0)]=[1.1,-2.2,-0.6,0.5]$, respectively.
We set the drive system in a chaotic supra-threshold bursting regime by fixing the values of its parameters at $a = 3.0$, $b = 1.0$, $\alpha = 0.1$, $\beta = 0.02$, $c = 1.0$, $d = 5.0$, $\sigma = 0.0278$, $\theta = 0.006$, $s=4.75$  $x_0 = -1.56$, $y_0 = -1.619$, $\mu = 0.0009$, $\gamma = 3.0$, $\rho = 0.9573$, $I = 3.1$,  $k_1=0.85$, and $k_2=0.5$. The response system is also set into a supra-threshold (without controllers $\mathbf{U} = (0,0,0,0))$ bursting regime by fixing the values of its known parameters at $\alpha = 0.1$, $\beta = 0.02$, $c = 1.0$, $s=4.75$ , $x_0 = -1.56$,  $I = 3.1$,  $k_1=0.85$, and $k_2=0.5$. 
The initial value of estimate for the ``unknown" parameters are $[\overline{a}(0),\overline{b}(0),\overline{d}(0),\overline{\theta}(0)]=[0,0,0,0]$. Furthermore, the initial values of the feedback gain parameters of the controllers are set at $[g_x(0),g_y(0),g_z(0),g_{\phi}(0)]=[0.5,0.5,0.5,0.5]$. 
The error variables are initialized  at $[e_x(0),e_y(0),e_z(0),e_{\phi}(0)]=[2.0,-2.0,2.0,-2.0]$.

The fourth-order Runge-Kutta algorithm is employ\-ed to simultaneously integrate the drive, response, and controlled error systems described by Eqs. \eqref{eq:master}, \eqref{eq:response}, and \eqref{eq:errordynamic_update}, respectively. With controllers $\mathbf{U}=(U_x,U_y,U_z,U_{\phi})^T$ chosen as in Eqs. \eqref{t1:1} and \eqref{t1:2}. Figures \ref{fig:ds_sim} \textbf{(a)}-\textbf{(c)} display the time evolution of the synchronization errors, the controllers' feedback gain parameters of the controllers, and the ``unknown" parameters' estimation. These numerical simulations show that the synchronization errors asymptotically converge to zero ($e_x\to0$, $e_y\to0$, $e_z\to0$, $e_{\phi}\to0$), the feedback gain parameters of the controllers converge to constant values ($g_x\to 2.72$, $g_y\to 10.22$, $g_z\to 2.05$, $g_{\phi}\to 2.00$), and the estimations of the ``unknown" parameters converge also to some constants ($\overline{a} \rightarrow  a=3.0$, $\overline{b} \rightarrow  b=1.0$, $\overline{d} \rightarrow  d=5.0$, $\overline{\theta} \rightarrow  \theta=0.006$), i.e., the two mismatched systems successfully achieve reduced-order adaptive synchronization. To clearly show the variation in errors, feedback gain, and estimated parameters from the initial conditions to the final states, we plotted the time axes in Figs. \ref{fig:ds_sim} \textbf{(a)}-\textbf{(c)} on a logarithmic scale. This adjustment was necessary due to the extended duration of the simulations.

\begin{figure}[t]
\centering
    \includegraphics[width=10.0cm,height=5.0cm]{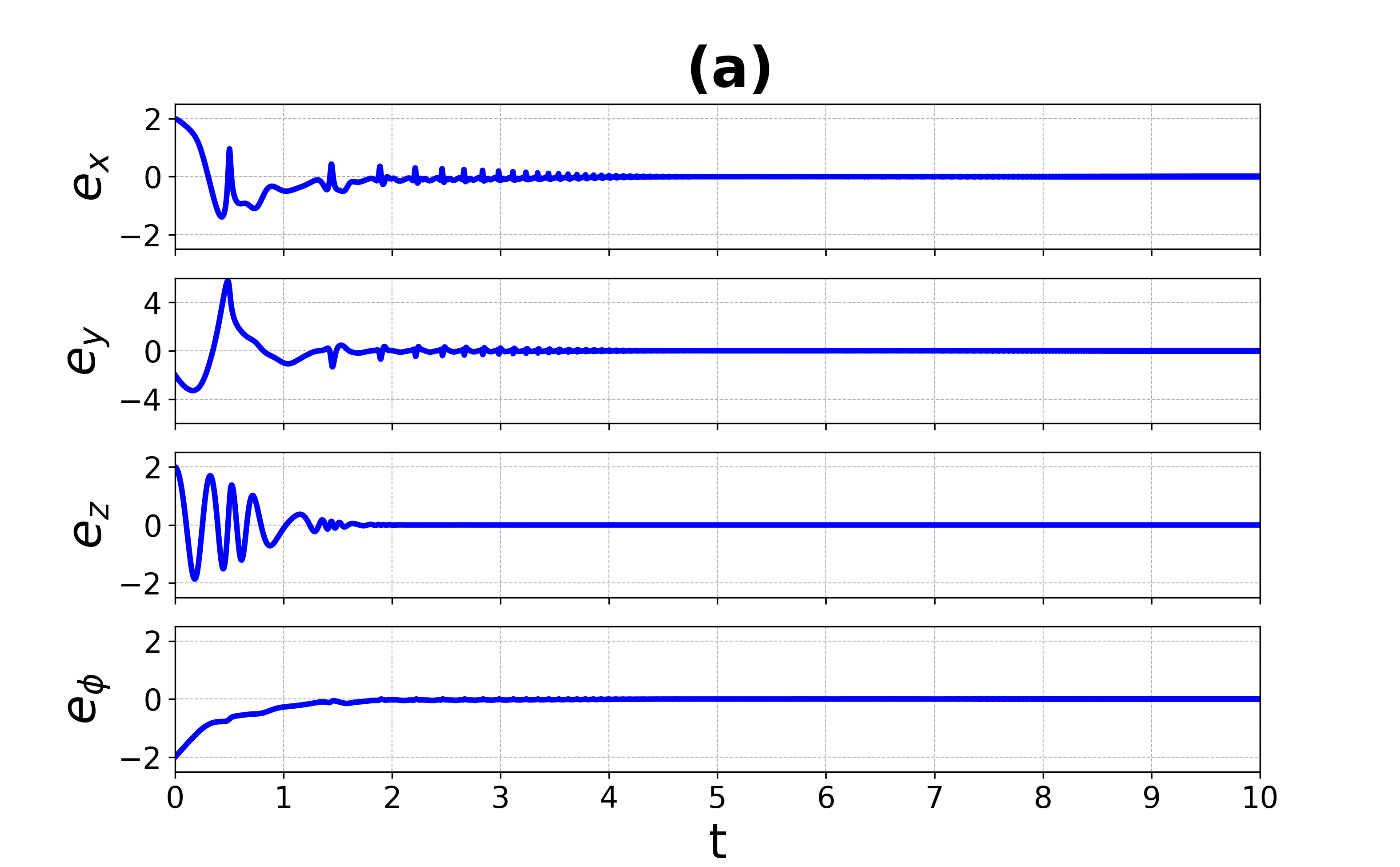}
    \includegraphics[width=10.0cm,height=5.0cm]{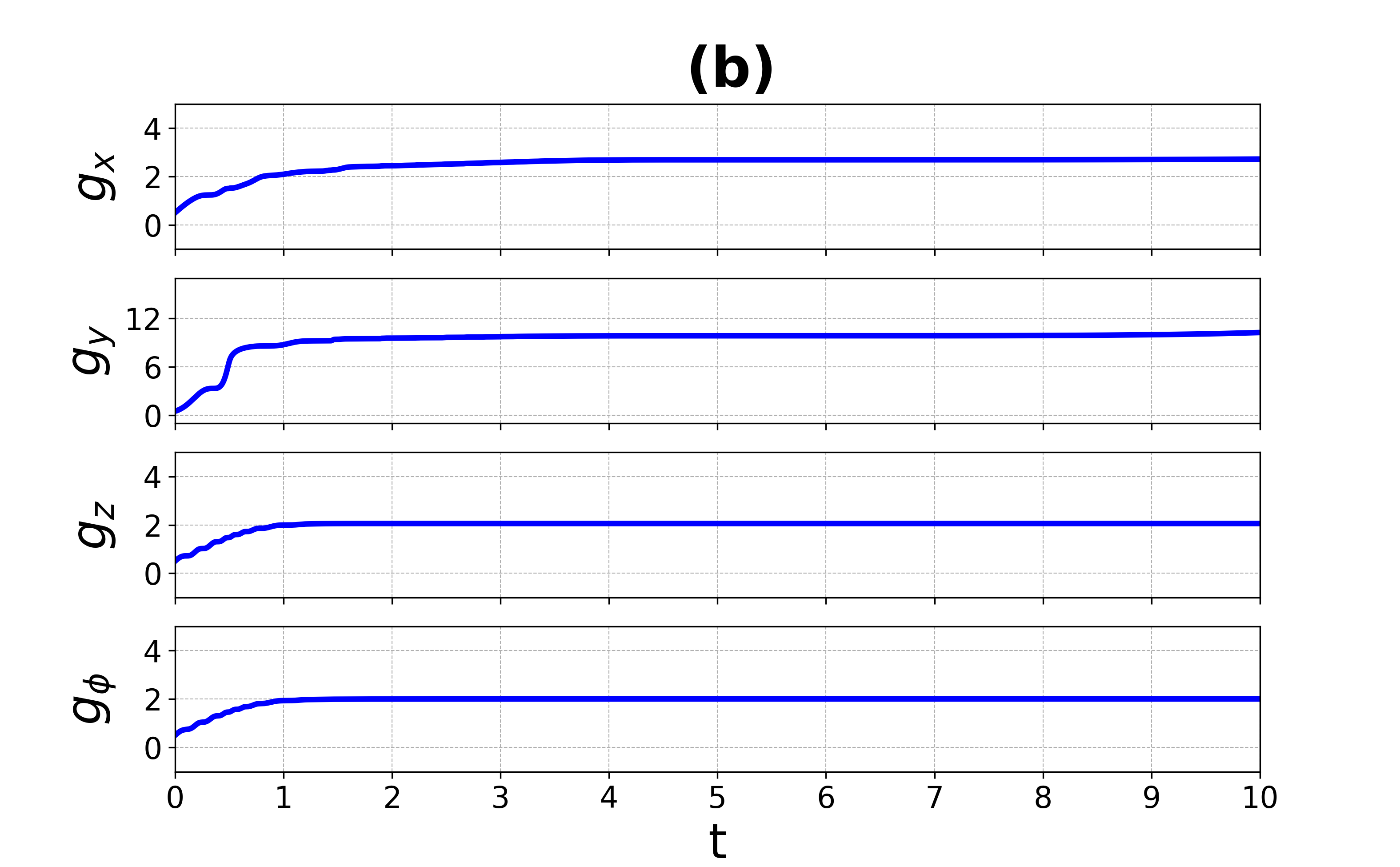}
\includegraphics[width=10.0cm,height=5.0cm]{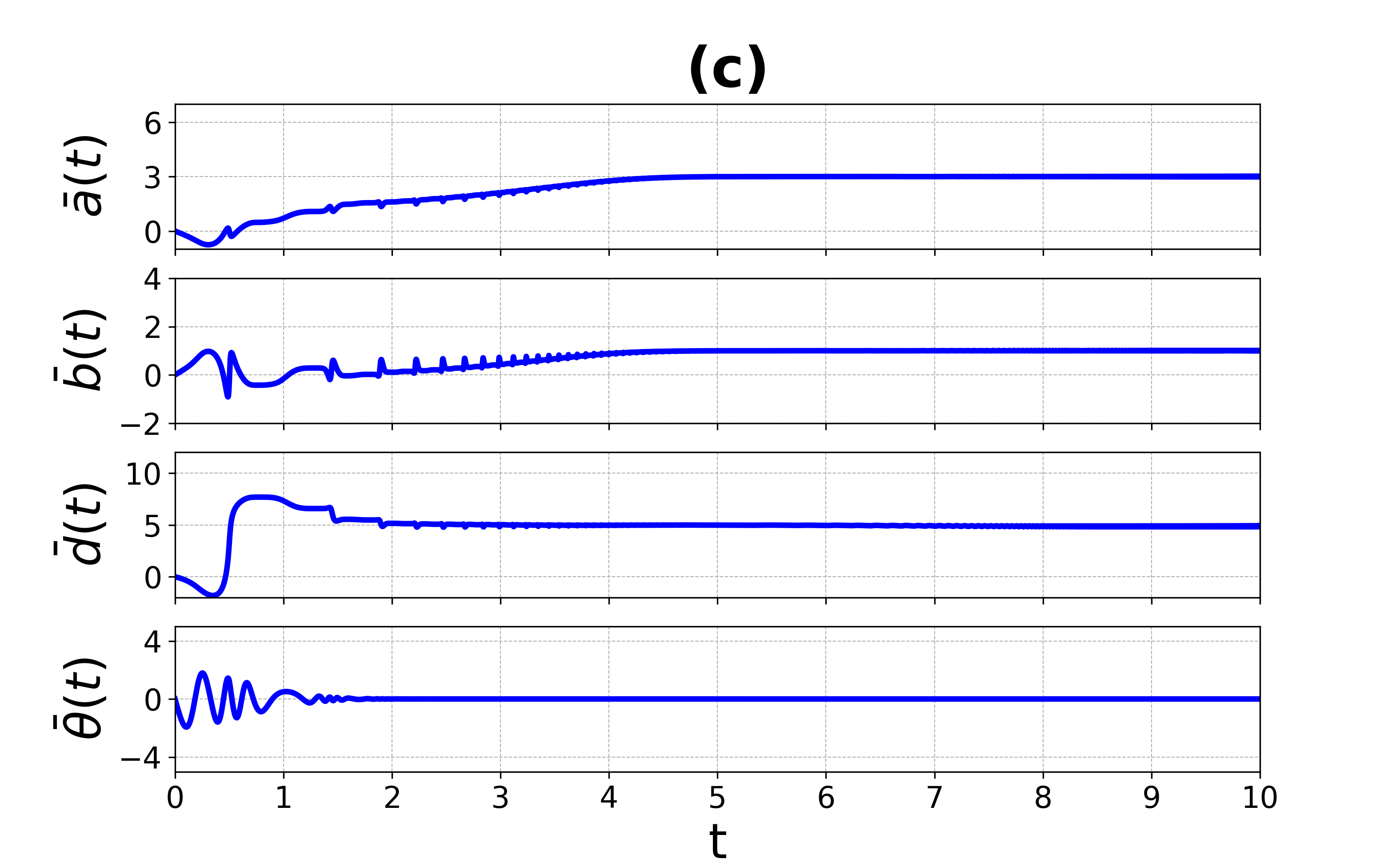}
\caption{Time series of the errors in \textbf{(a)}, the feedback gain parameters in \textbf{(b)}, and the adaptive parameters in \textbf{(c)}. Parameter values of the drive and response systems: see main text.}\label{fig:ds_sim}
\end{figure}

In Fig. \ref{fig:ds_k_1_k_2}, we investigate the stability of the synchronization manifold in the $k_1-k_2$ parameter space with a $200 \times 500$ grid points. The red region represents parameter values at which the reduced-order adaptive synchronization manifold  $\mathcal{M}_s$ is unstable, i.e., since $\dot{V}>0$, and the green region represents parameter values at which $\mathcal{M}_s$ is stable, i.e., since $\dot{V}=0$. This result indicates that the larger values of the magnetic gain parameter $k_2$, which describes the degree of polarization and magnetization of the neuron, destabilize the reduced-order synchronization, especially at smaller magnetic gain $k_1$. 

\begin{figure}[H]
    \centering
\includegraphics[width=8.0cm,height=5.0cm]{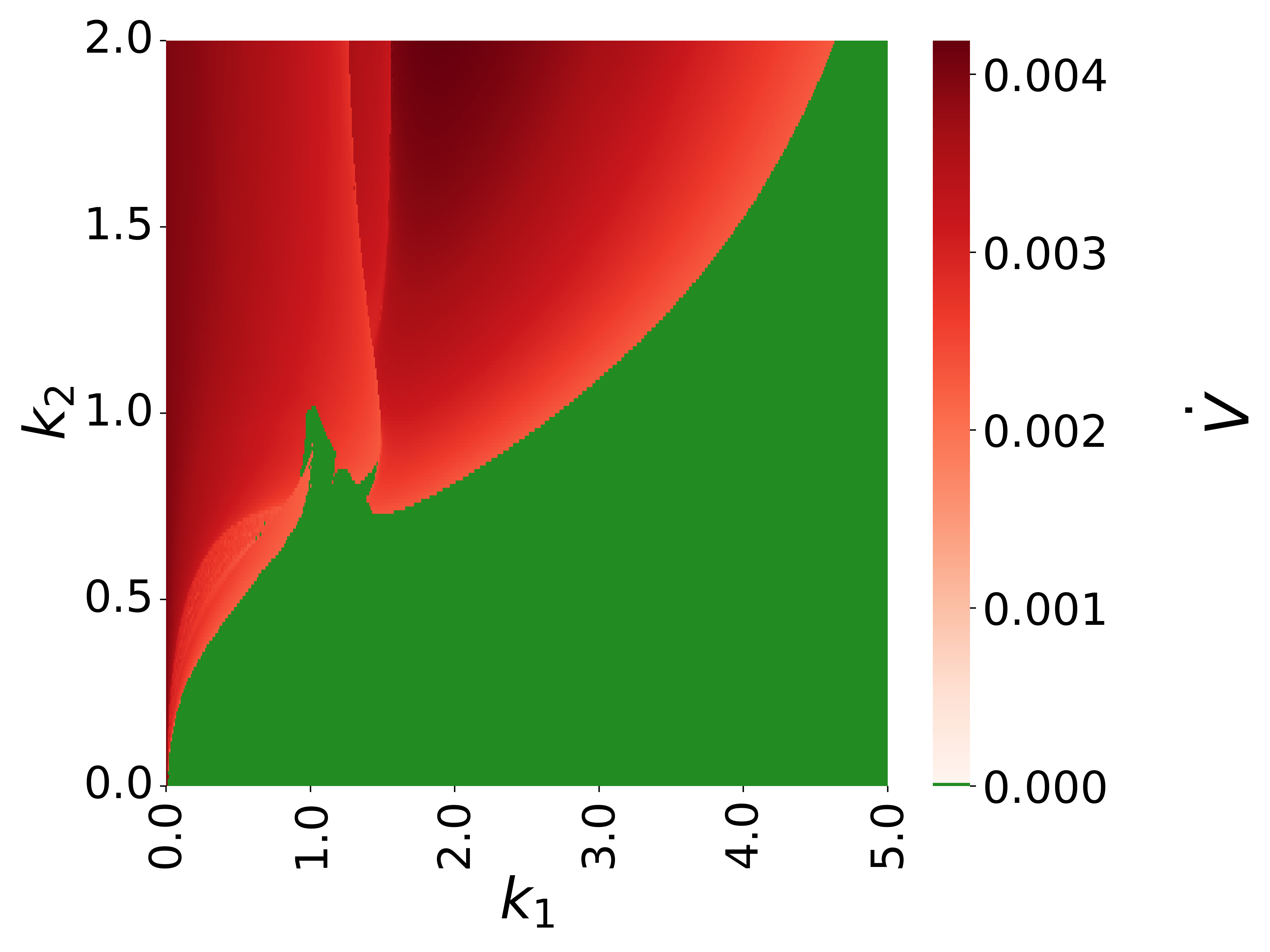}
    \caption{Variation of $\dot{V}$ in Eq. \eqref{eq:32} with the gain feedback parameters $k_1$ and $k_2$. Reduced-order adaptive synchronization is stable (unstable) in the green (red) region where $\dot{V}=0$ ($\dot{V}>0$). Here, $l_x = l_y = l_z = l_{\phi} = 1.0$. Other parameter values: see the main text.}\label{fig:ds_k_1_k_2}
\end{figure}

\clearpage
\section{Machine learning approach to the reduced-order synchronization problem}\label{sec:ML}
In this section, we address the problem of reduced-order synchronization through the lens of reservoir computing (RC) based on echo state networks (ESNs). 
\subsection{Standard ESN implementation}
In reservoir computing, an RNN serves as a fixed, random, large-scale dynamical system known as the reservoir which processes time-dependent or sequential signals, making it especially effective for tasks involving memory or temporal sequences. In the sequel, we provide a concise overview of the standard RC implementation. The standard framework of ESN (see Fig. \ref{fig:esn_diagram}) \cite{Jaeger:2007,Jaeger_Haas_2004} consists of three main components (i) the input layer with fixed random weights $\mathbf{W_{\text{in}}}\in\mathbb{R}^{m\times n}$, which receives the input data; (ii) the reservoir characterized by internal weights $\mathbf{W_{\text{res}}}\in\mathbb{R}^{m\times m}$ with a specified sparsity (typically $m \gg n$); and (iii) the output layer with a weight matrix $\mathbf{W_{\text{out}}}$, which is usually trained using regularized linear regression.

The RNN in the reservoir uses the dynamics of the reservoir states $\mathbf{r}(t) \in \mathbb{R}^m$, with initial conditions $\mathbf{r}(0)=0$, as a dynamic memory for past inputs, and  updates according to:
\begin{equation}\label{eq:ESN}
\mathbf{r}(t) = (1 - \alpha)\mathbf{r}(t-1) + \alpha \tanh \big[\mathbf{W_{\text{res}}} \cdot \mathbf{r}(t-1)  + \mathbf{W_{\text{in}}} \cdot \mathbf{u}(t)\big],
\end{equation}
where $\alpha\in(0,1]$ is the leaky coefficient (which controls how much the previous reservoir state $ \mathbf{r}(t-1) $ influences the next state $\mathbf{r}(t)$. A small $ \alpha $ makes the system more memory-driven, while a large $ \alpha $ makes it more input-driven). The nonlinear activation function is represented by $\tanh[\cdot]$, and $\mathbf{u}(t)\in\mathbb{R}^n$ is the input data vector, coming from, e.g., experiments, where measurements (data) are collected, or the numerical simulations (integration) of a dynamical system 
\begin{equation}\label{eq:real_dynamics}
\frac{d\mathbf{u}(t)}{dt}=F(\mathbf{u}(t)),
\end{equation}
where $F:\mathbb{R}^{n}\to\mathbb{R}^{n}$, modeling behavior of a given system, e.g., our HR neuron given in Eq. \eqref{eq:master}.

\begin{figure}
    \centering
\includegraphics[width=10.0cm,height=5.0cm]{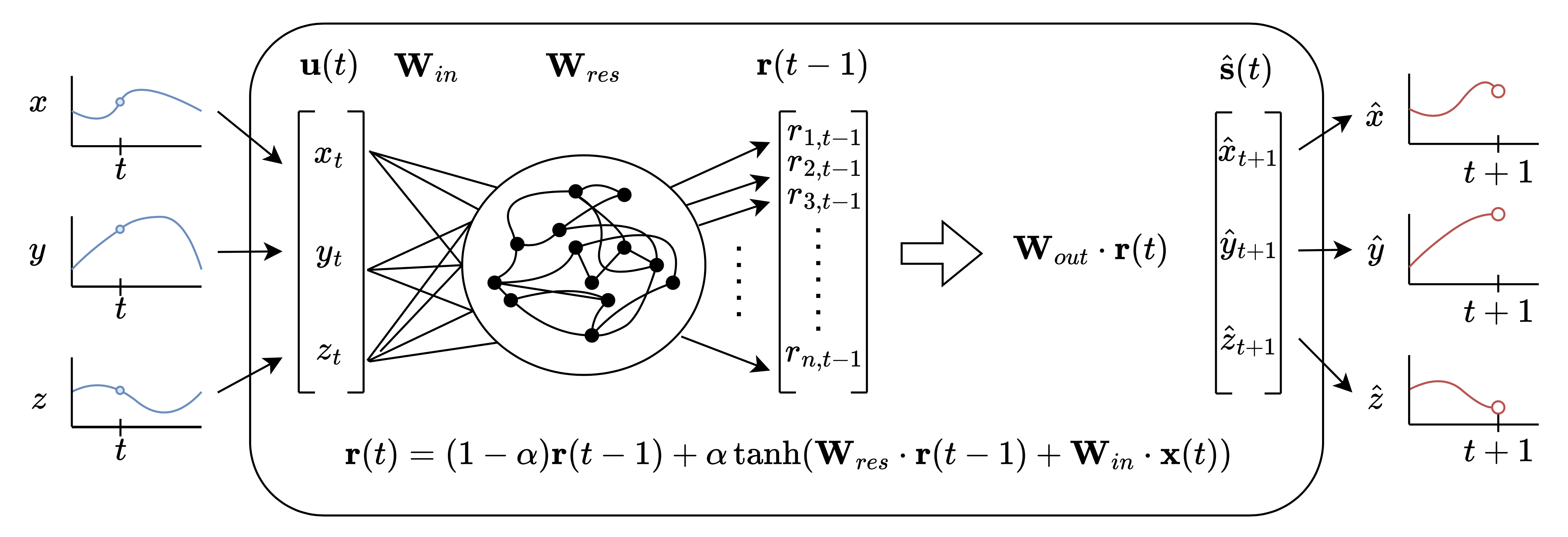}
    \caption{Block diagram illustrating the standard ESN architecture.}\label{fig:esn_diagram}
\end{figure}

The weight matrices $\mathbf{W_{\text{res}}}$ and $\mathbf{W_{\text{in}}}$ are randomly initialized matrices that are unchanged during the training and prediction phase. The fixed weights of $\mathbf{W_{\text{in}}}$ are randomly and uniformly sampled in  $[-1,1]$, and the adjacency matrix $\mathbf{W_{\text{res}}}$ of the reservoir has a specified sparsity and node degree, often uniformly distributed in  $[-1,1]$ and then adjusted by the spectral radius $p:= \max \{ |\lambda_i|: \lambda_i \text{ is an eigenvalue of } W_{\text{res}} \}>0$, which controls the ability of the reservoir to maintain the two desired properties of a good reservoir---echo state and computational capacity properties \cite{jaeger2004harnessing, Memory_versus_non-linearity}. In this paper, the weighted adjacency matrix $\mathbf{W_{\text{res}}}$ consists of random Erdős-Rényi without self-loops and a link probability of $0 < q\leq 1$.
The output $\hat{\mathbf{s}}(t)\in\mathbb{R}^{n}$ (see Fig. \ref{fig:esn_diagram}) of the ESN is given by:
\begin{equation}\label{eq:ml_esn_out}
\hat{\mathbf{s}}(t) = \mathbf{W}_{\text{out}} \cdot \mathbf{r}(t),
\end{equation}
where $\mathbf{W}_{\text{out}}\in\mathbb{R}^{n\times m}$ is the output weight matrix. Remarkably, unlike other neural networks, only $\mathbf{W}_{\text{out}}$ is trained in the standard ESN scheme.

To train the ESN model (i.e., to train $\mathbf{W}_{\text{out}}$), we collect true data given by the solutions of a dynamical system in Eq.\eqref{eq:real_dynamics} 
to generate $ N_{\text{train}} $ input states $\{\mathbf{u}(t)\}_{t=0,..., N_{\text{train}}}$. Next, we use these input states to drive the reservoir dynamics in Eq.\eqref{eq:ESN}, starting from $\mathbf{r}(0) = 0 $, and obtain the corresponding reservoir states $\{\mathbf{r}(t)\}_{t=0,..., N_{\text{train}}}$. We define the targets $\{\mathbf{s}(t)\}_{t=0,..., N_{\text{train}}}$ as $\mathbf{s}(t) = \mathbf{u}(t+1)$. We remove the first $N_{\text{transient}}$ columns of $\mathbf{S}$ and $\mathbf{R}$, where $\mathbf{S}$ and $\mathbf{R}$ are matrices with columns $\{\mathbf{S}(t)\}$ and $\{\mathbf{r}(t)\}$, for $t = N_{\text{transient}}, ...,N_{\text{train}}$, respectively \cite{Lukoševičius2012}. We then compute the output weights $\mathbf{W}_{\text{out}}$ to best fit the training data by solving the Ridge Regression problem with Tikhonov regularization:
\begin{equation}\label{eq:5}
 \mathbf{W}_{\text{out}} = \mathbf{S}\mathbf{R}^T (\mathbf{R}\mathbf{R}^T + \lambda_0\mathbb{I})^{-1},
\end{equation}
where $\mathbb{I}$  is identity matrix. The regularization parameter $ \lambda_0> 0 $ helps prevent over-fitting by mitigating ill-conditioning of the weights.

The trained model is subsequently used to predict the dynamics of Eq. \eqref{eq:real_dynamics} $N_{\text{pred}}$ steps into the future. In this process, the output of the model given by  Eq. \eqref{eq:ml_esn_out}, i.e., $\hat{\mathbf{s}}(t) = \mathbf{W}_{\text{out}} \cdot \mathbf{r}(t)$, in each time step serves as the new input $\mathbf{u}(t+1)$ in Eq. \eqref{eq:ESN} for the following step, for $N_{\text{train}} \leq t \leq N_{\text{pred}}$. 

 The root mean square error (RMSE) defined in  Eq. \eqref{eq:RMSE} is our measure of the prediction performance of the ESN model.
\begin{equation}\label{eq:RMSE}
RMSE= \sqrt{\dfrac{1}{N_{\text{pred}} - N_{\text{train}}}\sum_{t=N_{\text{train}}}^{N_{\text{pred}}}\big[\mathbf{s}(t) - \hat{\mathbf{s}}(t)\big]^2},
\end{equation}
where $\mathbf{S}$ and  $\hat{\mathbf{S}}$ are matrices with columns $\{\mathbf{s}(t)\}$ containing the true dynamics of the system in Eq. \eqref{eq:real_dynamics} and columns $\{\hat{\mathbf{s}(t)\}}$ containing the output (predicted data) of the system for $t = N_{\text{train}}, ..., N_{\text{pred}}$, respectively. 

Tuning the hyper-parameters $\{m, \alpha, q, p, \lambda_0\}$ is essential for achieving good prediction performance. We utilize a Bayesian hyper-parameter optimization algorithm \cite{ESN_bayes} from the scikit-optimize library \cite{skopt} to obtain an ESN that achieves a small RMSE on the prediction evaluation dataset.

In our data-driven approach to the reduced-order synchronization problem, we generate input datasets, namely, $[x(t),y(t),z(t),w(t),\phi(t)]^T=\mathbf{u}(t)\in \mathbb{R}^5$ and $[x_r(t),y_r(t),z_r(t),\phi_r(t)]^T=\mathbf{u}_r(t)\in \mathbb{R}^4$ from the drive and response of the 5D and 4D HR neuron models, respectively. This is done by numerically integrating Eqs. \eqref{eq:master} and \eqref{eq:response} without controllers ($U = 0$) using the fourth-order Runge-Kutta algorithm with a step size of $dt = 0.1$ and for a total time step of $N = 22000$, from which we discarded the first $20000$ transient solutions. We sample the training data at $2dt$ step. We set the bifurcation parameters to $k_1 = 0.21$ and $k_2 = 0.4$, with other parameters as previously defined, resulting in a chaotic supra-threshold bursting time series in the 5D and 4D HR neuron models.

\begin{figure}
    \centering
\includegraphics[width=11.0cm,height=5.0cm]{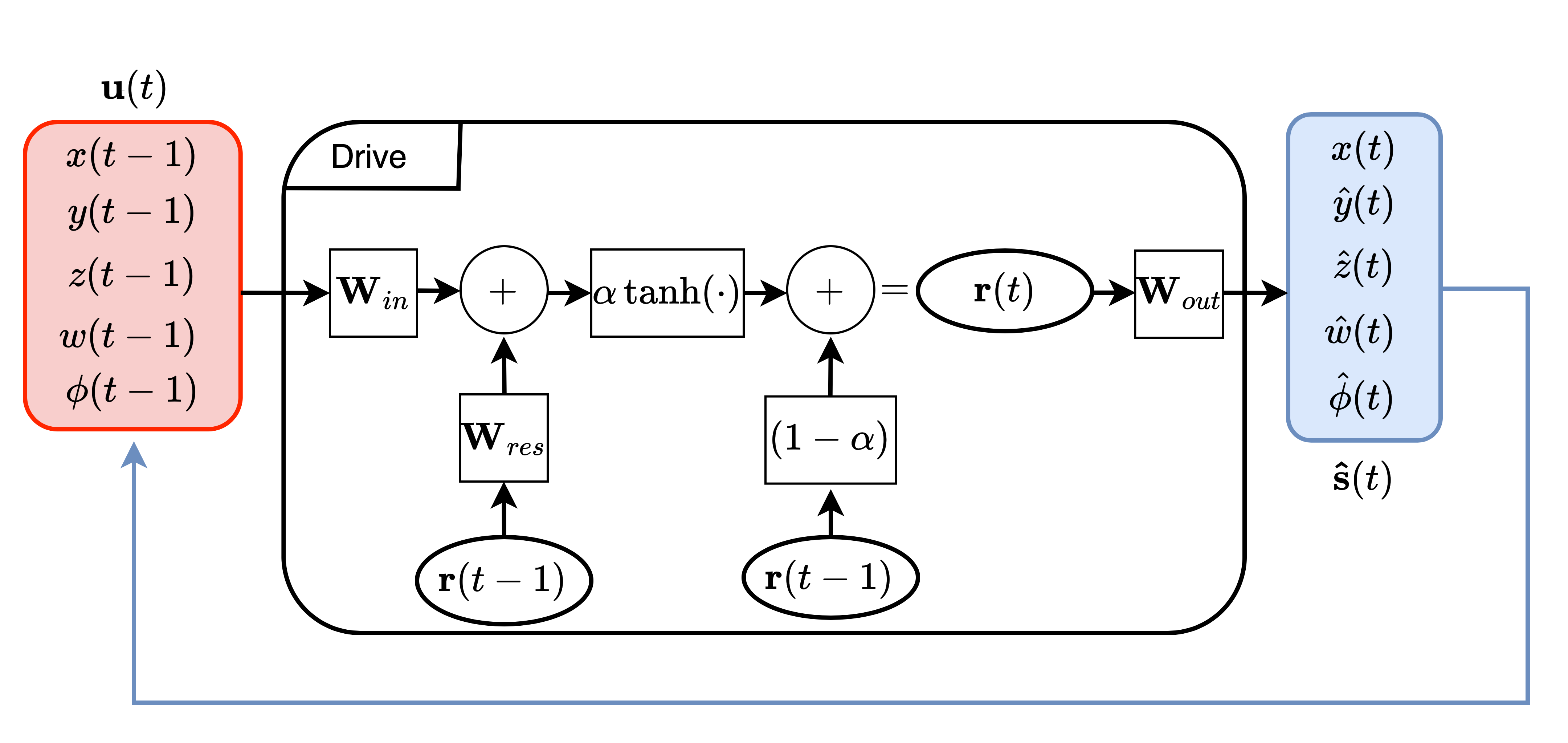}
    \caption{Block diagram illustrating the standard ESN architecture for the drive 5D HR neuron. During the training phase, the data in the red box serves as input to the ENS. During the prediction phase, the output data in the blue box replaces the red box and becomes input data and the ENS becomes autonomous.}\label{fig:esn_block-diagram}
\end{figure}

We set $N_{transient}$=300 and train the systems with the time series in the interval $[20000, 20600]$  and use the time series in the interval $[20601, 21000]$ to find the best hyper-parameters of the ESN, i.e., the ESN (see in Fig. \ref{fig:esn_block-diagram}) is used to predict the time series of the drive system given by Eq. \eqref{eq:master}. We train the models with the optimized hyper-parameters to predict the dynamics in the interval $[20601,21200]$. 

Figure \ref{fig:esn} \textbf{(a)} shows the performance of the drive ESN (red) and the training/evaluation data of Eq. \eqref{eq:master} (blue). The errors between the two time series can be seen in Fig. \ref{fig:esn} \textbf{(b)}. The drive ESN achieves an RMSE of \num{4.19e-2}. However, due to the chaotic nature of time series data \cite{Ramadevi2022}, prediction accuracy quickly declines over time, leading to a rapid and significant widening of error bounds as the prediction horizon extends.

This divergence indicates that, over time, the attractors predicted by our standard ESN would deviate from the true attractors of the 4D and 5D chaotic HR neuron models. Our goal is to adaptively achieve asymptotic synchronization (in reduced order) of the accurately predicted time series of these neuron models, rather than attempting to synchronize (which is still possible) time series that have deviated from their true trajectories. One approach to this problem is the reservoir observer technique.

\begin{figure}[t]
    \centering
        \includegraphics[width=12.0cm,height=6cm]{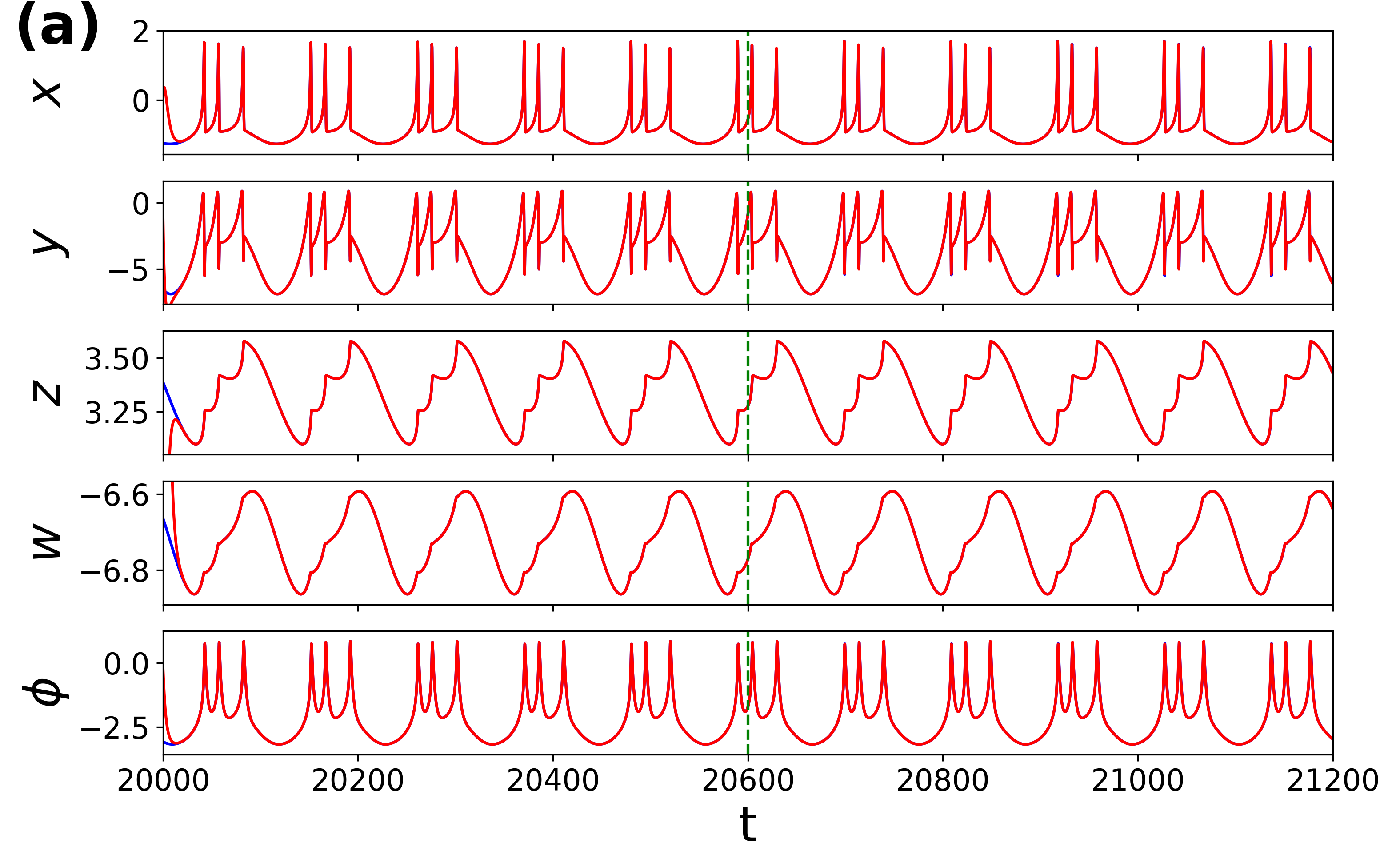}
        \includegraphics[width=12.0cm,height=6cm]{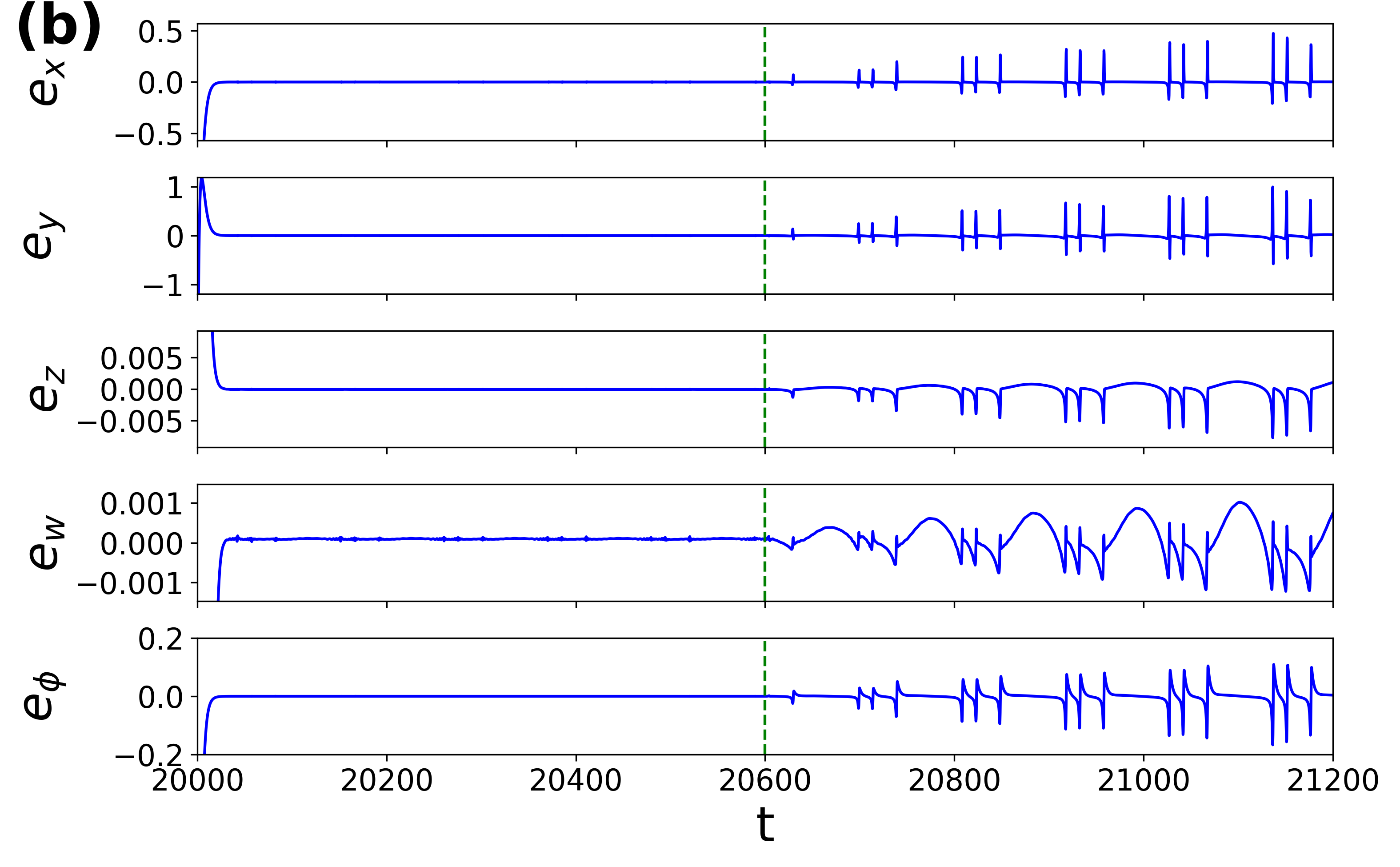}        
    \caption{\textbf{(a)} Time series of the 5D HR neuron (blue) and ESN time series of the 5D HR neuron  (red). \textbf{(b)} Error between the two time series in \textbf{(a)} grows with time. The green vertical line in \textbf{(a)} and \textbf{(b)} separate the training phase (left) from the prediction phase (right). The 4D HR neuron has a similar behavior (figures not shown).}\label{fig:esn}
\end{figure}



\subsection{ Reservoir observer (RO)}
A reservoir observer (RO) is a technique for estimating a system's state when only partial observations are available continuously; however, a limited-time series dataset containing both the observed and later unobserved variables exists. The RO is trained and evaluated on the limited data set where the input is the observed variables, and the output is the later unobserved variables. The reservoir of the RO continually receives measured variables as input and generates estimates of the unmeasured variables, effectively reconstructing the system's state under the condition of observability \cite{lu2017reservoir}.  
In this section, we employ the RO method to enhance the accuracy and extend the prediction time of our chaotic attractors.

 In the RO paradigm, one or more variables from the input data are continuously observed and are known in advance for the entire duration of the time series, so they do not need to be predicted. From the HR models in Eqs. \eqref{eq:master}  and \eqref{eq:response}, we choose $x(t)$ and $x_r(t)$ as our observed variables. From a neuroengineering perspective, this choice is justified because measuring the membrane potential of neurons is considerably easier using well-established techniques like patch-clamping (invasive) or EEG (non-invasive), which give us full access to the membrane potential variables $x(t)$ and $x_r(t)$. In contrast, tracking the movement of fast and slow ions across the membrane requires more complex and specialized methods, such as ion-specific fluorescent indicators and calcium imaging, giving us little access to the variables $[y(t), z(t), w(t), \phi(t)]$ and  $[y_r(t), z_r(t), \phi_r(t)]$ which we will need to predict.
 
\begin{figure}[H]
    \centering
\includegraphics[width=10.0cm,height=10cm]{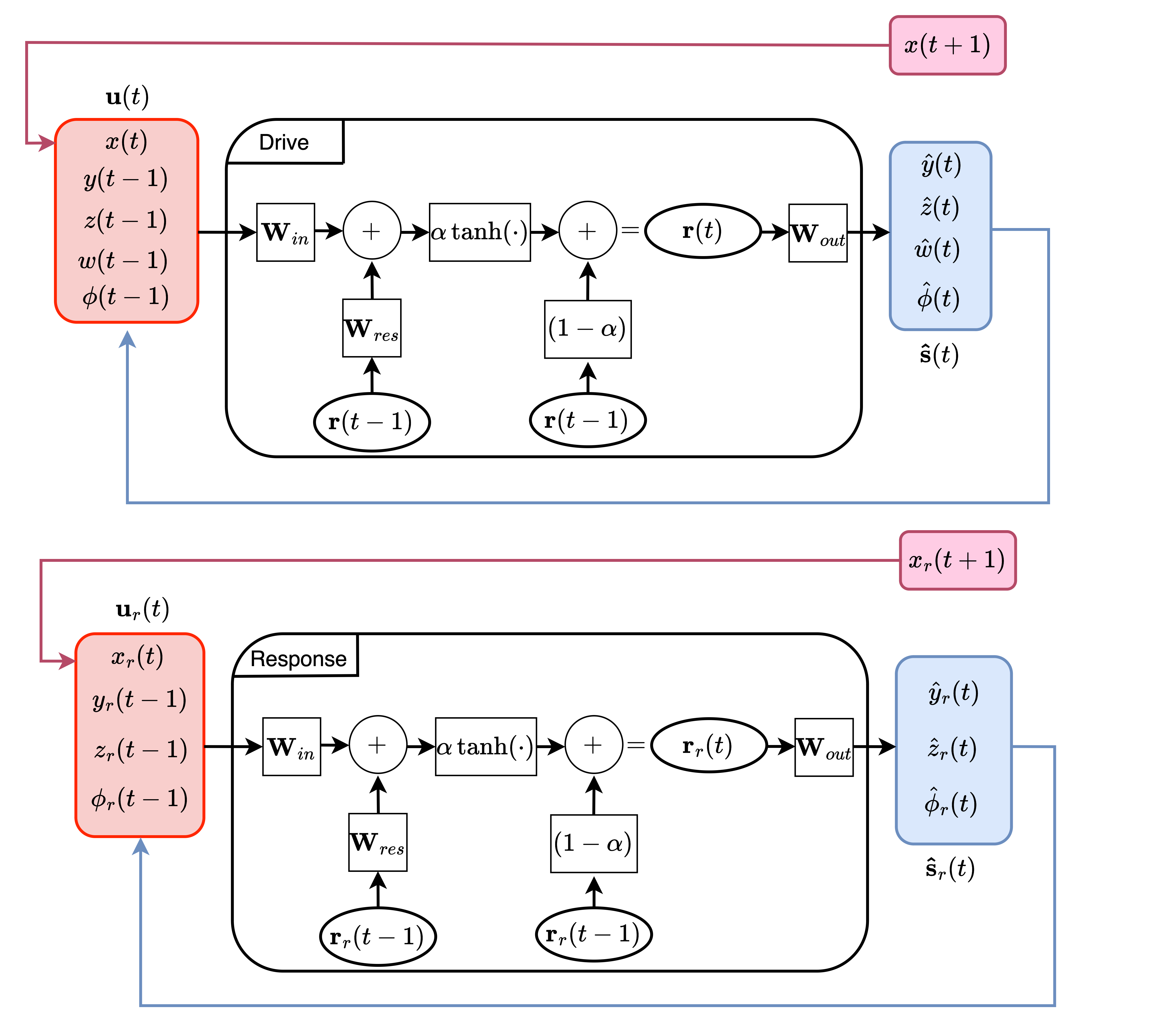}
    \caption{Block diagrams illustrating the RO architectures for the 5D HR neuron (top) and the 4D HR neuron (bottom). In the prediction phase, the output data in the blue box concatenates with the new observed variable [$x(t+1)$ or $x_r(t+1)$] and replaces the red box to become the input.}\label{fig:esn_observer_diagra}
\end{figure}

In Fig. \ref{fig:esn_observer_diagra}, we illustrate the architectures of the RO for the 5D and 4D HR models in a block diagram. We adjust the outputs $\hat{\mathbf{s}}(t) = [y(t), z(t), w(t), \phi(t)]$ and $\hat{\mathbf{s}}_r(t) = [y_r(t), z_r(t), \phi_r(t)]$, which represent the unknown variables that need to be predicted using ESN, at the current time step. To incorporate controllers for all unknown variables, we concatenate the observed variables with the previous outputs to create the input vectors $\mathbf{u}(t)$ and $\mathbf{u}_r(t)$.

\begin{figure}[H]
    \centering
        \includegraphics[width=12.0cm,height=5.5cm]{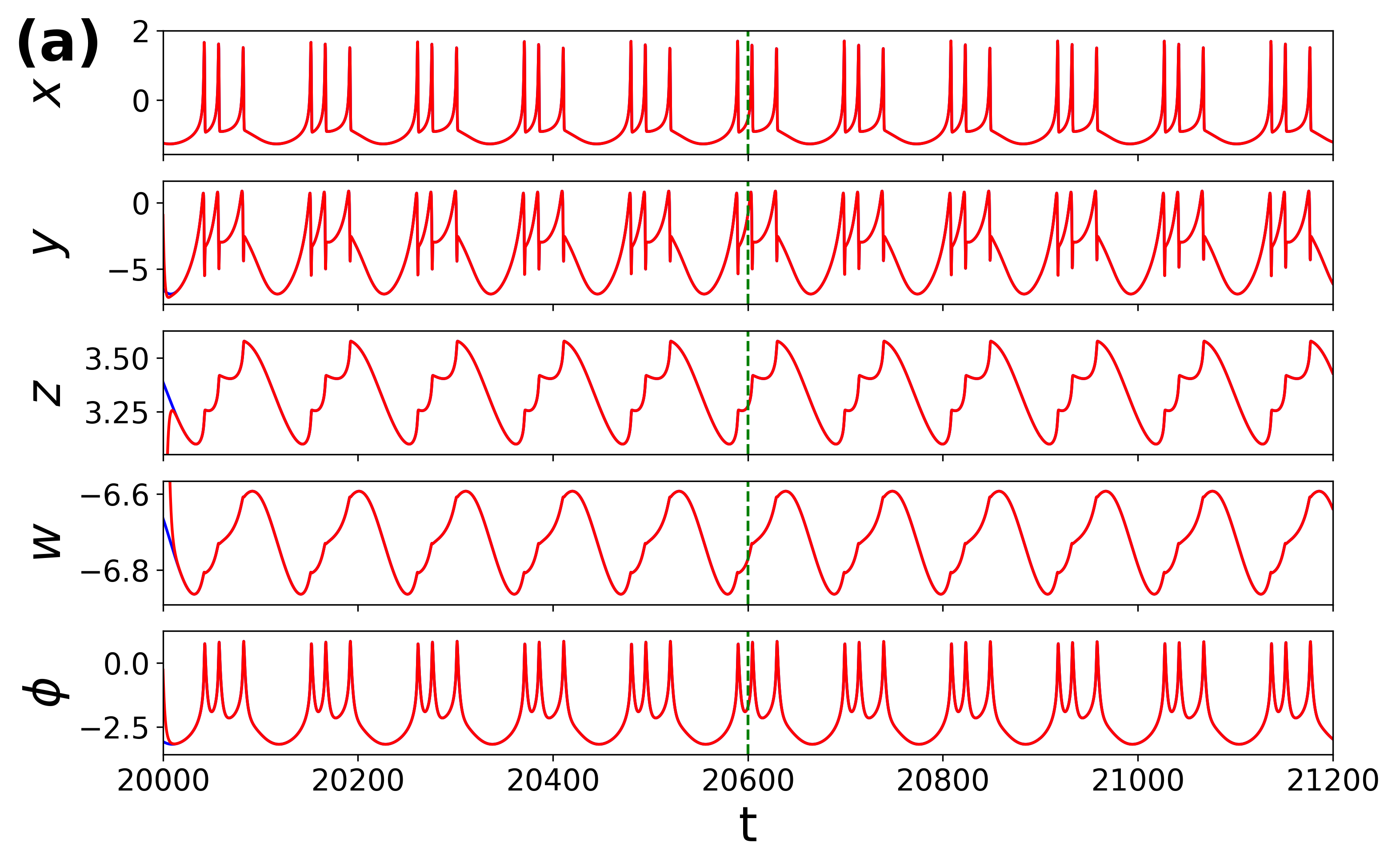}
        \includegraphics[width=12.0cm,height=5.5cm]{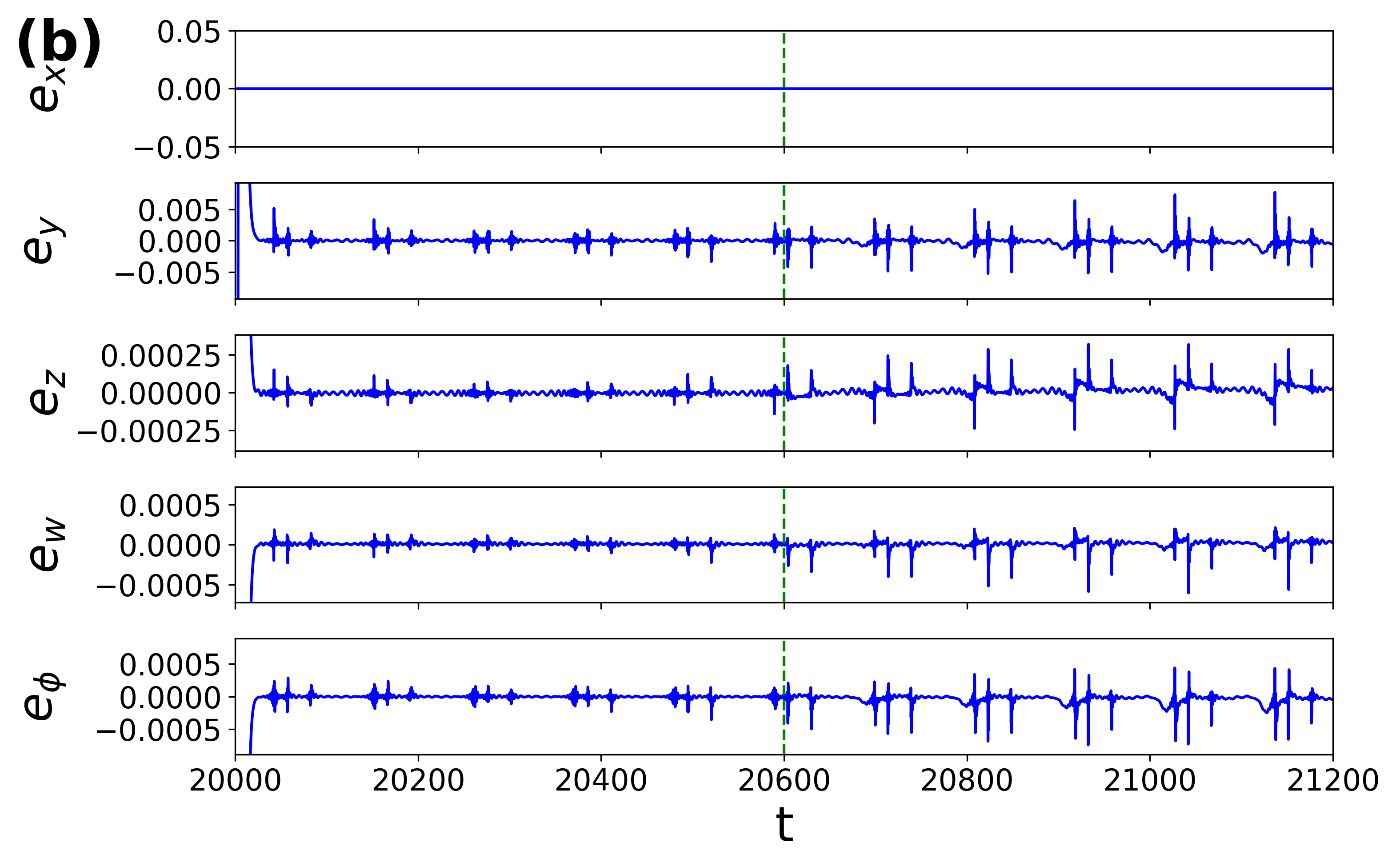}      
        \includegraphics[width=12.0cm,height=5.5cm]{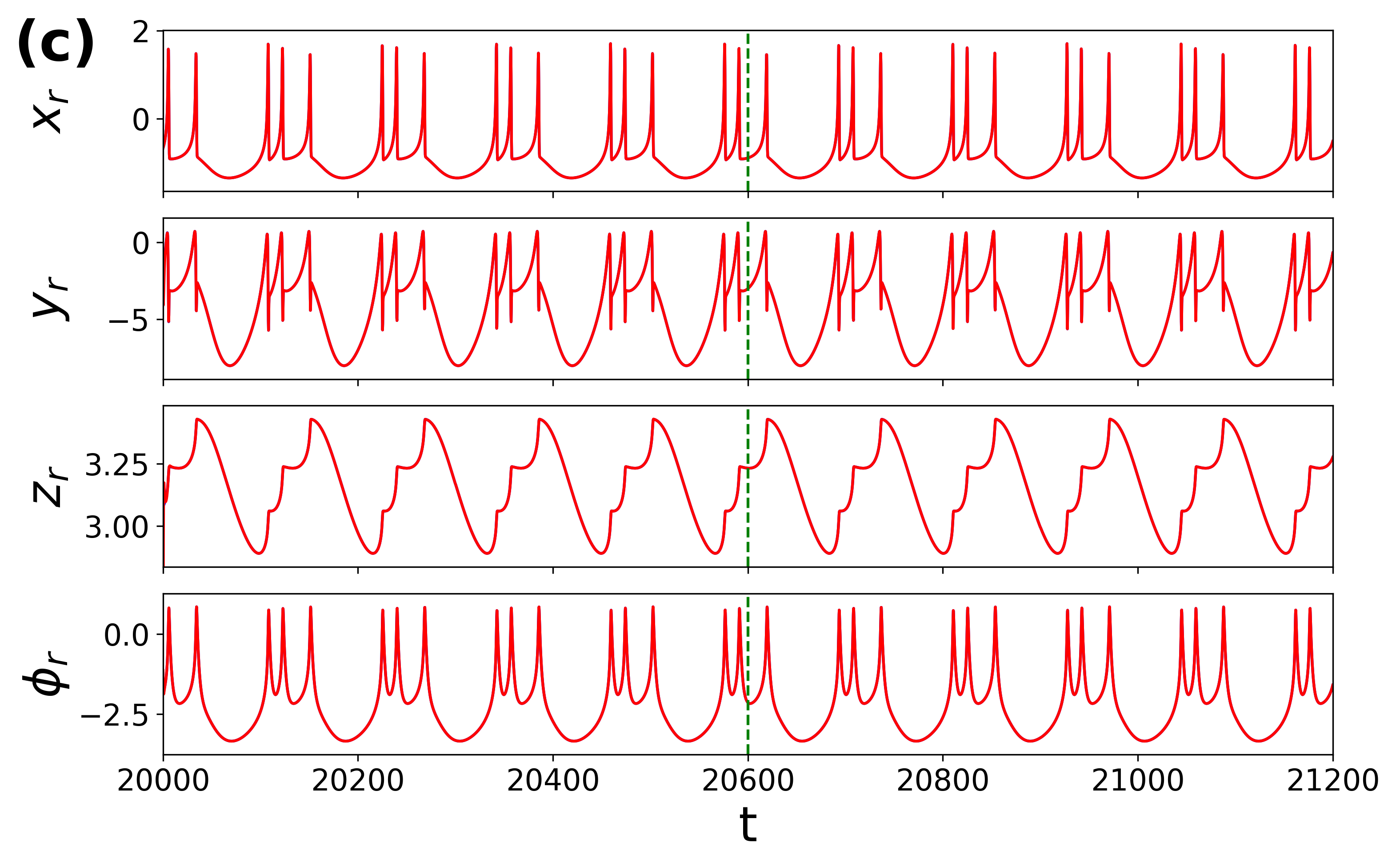}
        \includegraphics[width=12.0cm,height=5.5cm]{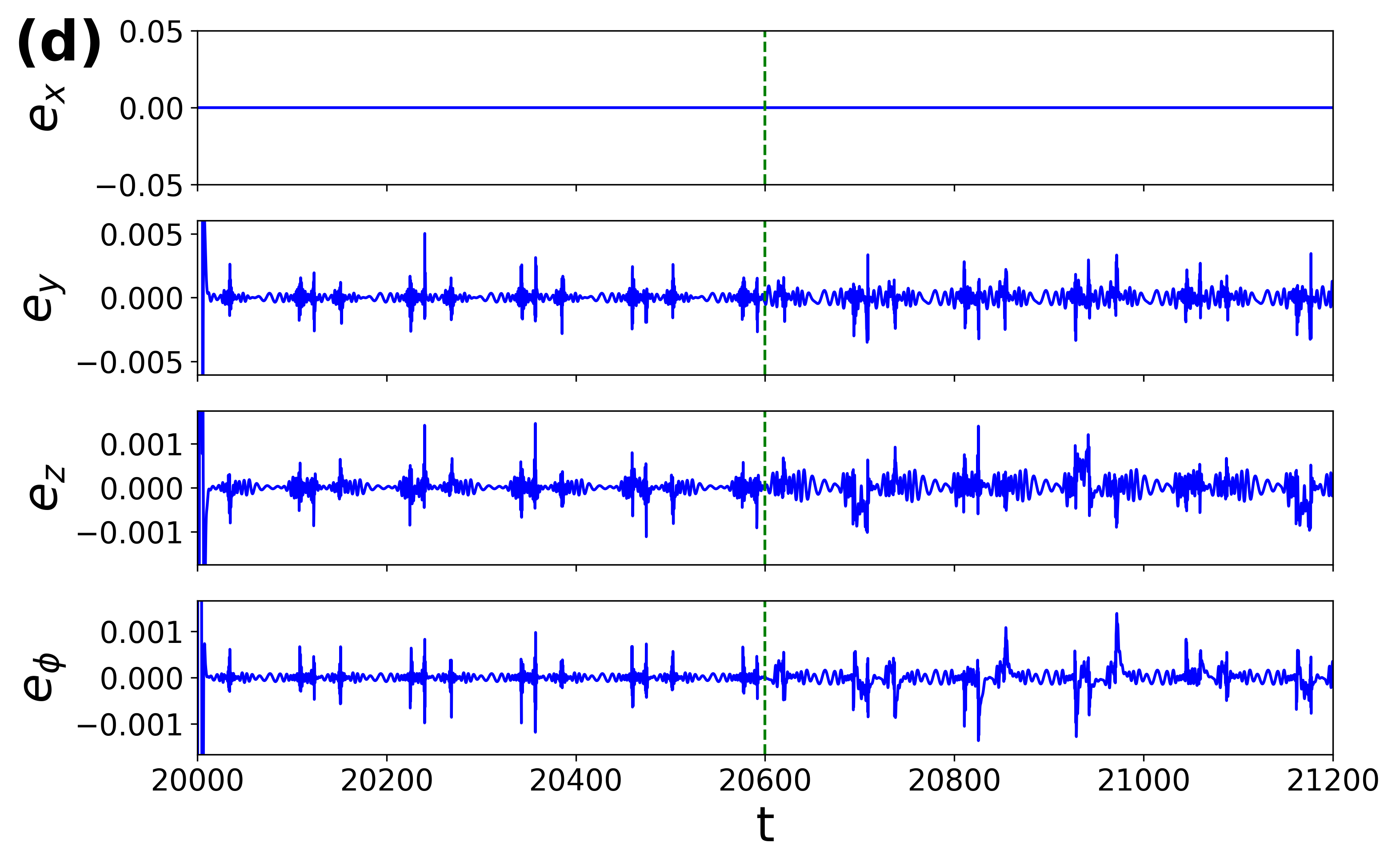}     
    \caption{\textbf{(a)} Time series of the 5D HR neuron (blue) and RO time series of the 5D HR neuron (red). 
    \textbf{(b)} Errors between the two chaotic time series in (a) are smaller for a significant time.
    \textbf{(c)} Time series of the 4D HR neuron (blue) and RO time series of the 4D HR neuron (red).  \textbf{(d)} Errors between the two time series in (c) are small and do not grow with time. 
The green vertical lines in \textbf{(a)}, \textbf{(b)}, \textbf{(c)}, and \textbf{(d)},  separate the training phase (left) from the prediction phase (right).}\label{fig:esn_observer}
\end{figure}

Figure \ref{fig:esn_observer}\textbf{(a)} and \textbf{(c)} illustrate the performance of the RO method in predicting the 5D drive and the 4D response HR neuron models (red) and the simulated time series (blue) for the training phase $[20000,20600]$ and prediction phase $[20601,21200]$. 
The errors between the time series of the 5D drive HR model and its RO-predicted time series are depicted in Fig. \ref{fig:esn_observer}\textbf{(b)}, where the drive RO achieves an RMSE of \num{3.4e-4}. Similarly, Fig. \ref{fig:esn_observer}\textbf{(d)} shows the errors between the time series of the 4D response HR model and its RO-predicted time series with a RMSE of \num{3.2e-4}.

Notice in Fig. \ref{fig:esn_observer} that the errors associated with the variables $x(t)$ and $x_r(t)$ are exactly zero since they are observed directly in the RO. Comparing Figs. \ref{fig:esn} and \ref{fig:esn_observer}, it is evident that the RO models outperform the standard ENS models significantly. Specifically, the errors with the RO technique are consistently smaller and remain low over time, whereas the errors from the standard ENS model are relatively larger and tend to increase over time.



\subsection{Adaptive control of reduced-order synchronization}

We address the reduced-order synchronization problem of 4D and 5D HR neurons by synchronizing the two data-driven ROs depicted in Fig. \ref{fig:esn_observer_diagra}. To achieve this, we employ two control strategies: online control and online predictive control.
Since, $ x(t)$ and $ x_r(t)$ represent the observed variables of the drive and response ROs, respectively,  the control of the synchronization error, $e_x$, is straightforward because we can use $x(t)$ as the observed variable for both the drive and response ROs. In the following sections, we will examine the performance of these control algorithms in detail.

\subsubsection{Online control (OC) scheme}

The online control scheme is a method that allows a model to make real-time decisions and continuously ad\-apt based on streaming data or feedback from a changing environment. Unlike offline control, it updates continuously to handle evolving system behaviors and unforeseen events, making it valuable for achieving adaptive synchronization \cite{murphy_probml}.

 Having trained two RO models (see Fig. \ref{fig:esn_observer}) to predict their respective time series, we now focus on controlling the ROs to achieve reduced-order synchronization. The training/prediction and control phases are treated separately. During training and prediction, we employ the approach illustrated in Fig. \ref{fig:esn_observer_diagra}, while for control, we switch to the setup in Fig. \ref{fig:esn_c2}. In the control phase, we modify the input $\mathbf{u_r}(t)$ of the response RO, which comprises of the observed variable $x(t)$ and the previous output $\hat{\mathbf{s}}_r(t)$.\\ The error vector $\mathbf{e}(t) = [e_y(t), e_z(t), e_{\phi}(t)]^T$ is computed as the difference between the drive system's output $\hat{\mathbf{s}}(t)$ and the response system's output $\hat{\mathbf{s}}_r(t)$. The error is then used to control the response system's output, such that: $i(t) = i_r(t) + e_i(t)$ where $i \in \{y,z,\phi\}$. The control is illustrated in Fig. \ref{fig:esn_c2} by using $y(t), z(t)$, and $\phi(t)$ directly from $\hat{\mathbf{s}}(t)$.

We update $\mathbf{W}_{out}$ of the response RO according to
\begin{equation}\label{ml:eq_online}
    \tilde{\mathbf{W}}_{out}(t+1) = \tilde{\mathbf{W}}_{out}(t) - \lambda_1\mathbf{r}_r(t)\cdot \mathbf{e}(t),
\end{equation}
where $\lambda_1=0.001$ represents the learning rate, $\mathbf{r}_r(t)$ is the reservoir state of the response RO, and the error $\mathbf{e}(t) = [e_y(t), e_z(t), e_{\phi}(t)]^T$. The output weight $\tilde{\mathbf{W}}_{out}(t+1)$ of the response RO is updated whenever $\mathbf{e}(t)$ is non-zero. Fig. \ref{fig:esn_c2} shows a block diagram of the proposed online control algorithm, with the green box representing the online learning algorithm in Eq. \eqref{ml:eq_online}, responsible for updating $\tilde{\mathbf{W}}_{out}(t+1)$ in real-time.

\begin{figure}
    \centering
\includegraphics[width=10.0cm,height=10cm]{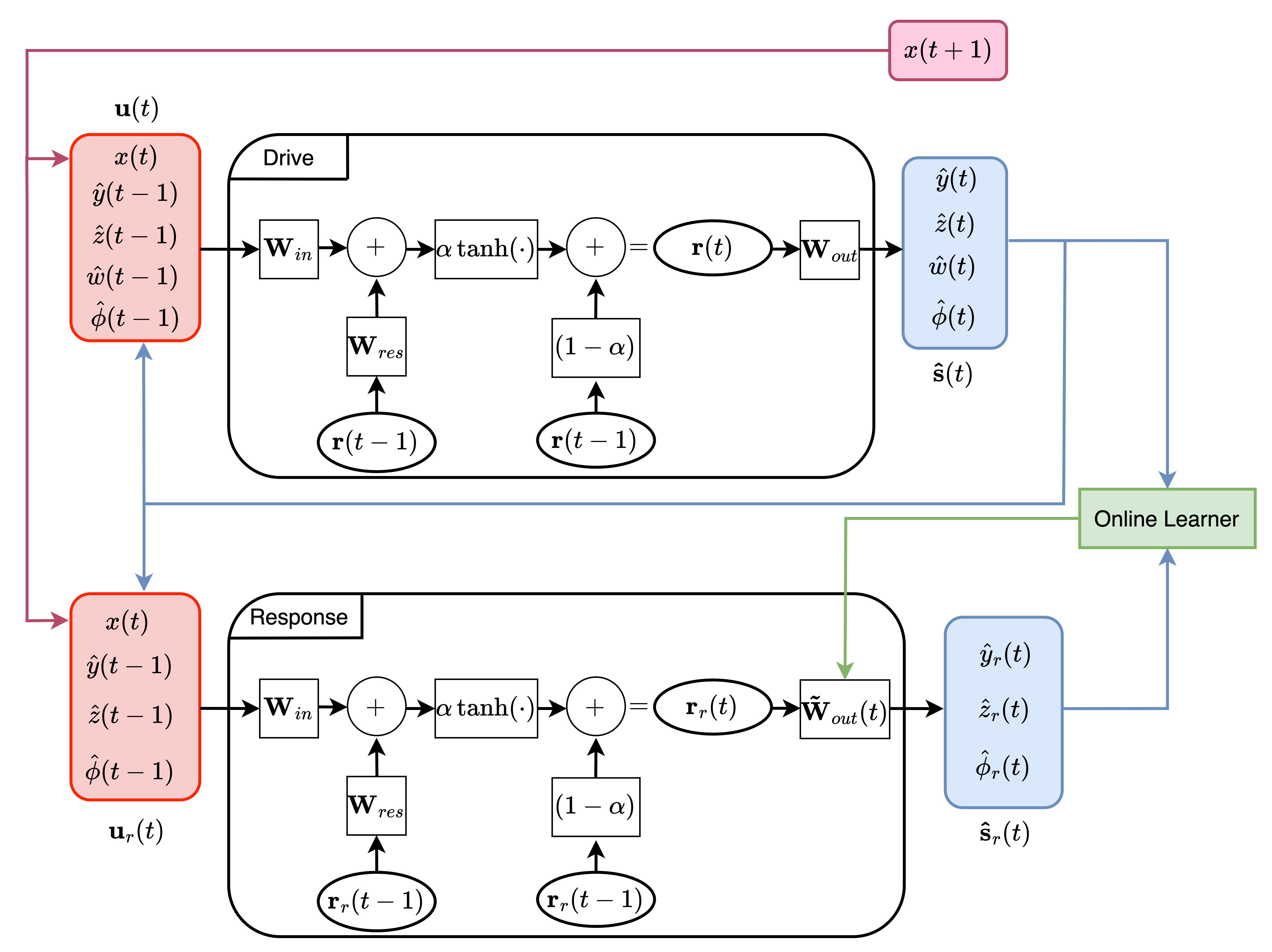}
    \caption{Block diagram illustrating the architecture during the control phase of the online control scheme.}\label{fig:esn_c2}
\end{figure}

Figure \ref{fig:esn_controller_1}\textbf{(a)} displays the training phase $[20000, 20600]$, prediction phase $[20600,21200]$, and the control phase $[21201,21800]$ of both ROs. The training and prediction phases correspond to those in Figs. \ref{fig:esn_observer}\textbf{(a)} and \textbf{(c)}. Upon application of the control algorithm, the time series synchronizes. Figure \ref{fig:esn_controller_1}\textbf{(b)} shows the error $\mathbf{e}(t) = [e_x(t),e_y(t), e_z(t), e_{\phi}(t)]^T$ in the control phase $[21200, 21800]$, which achieves a RMSE of \num{3.70e-2}.

\begin{figure}[H]
    \centering
        \includegraphics[width=12.0cm,height=6cm]{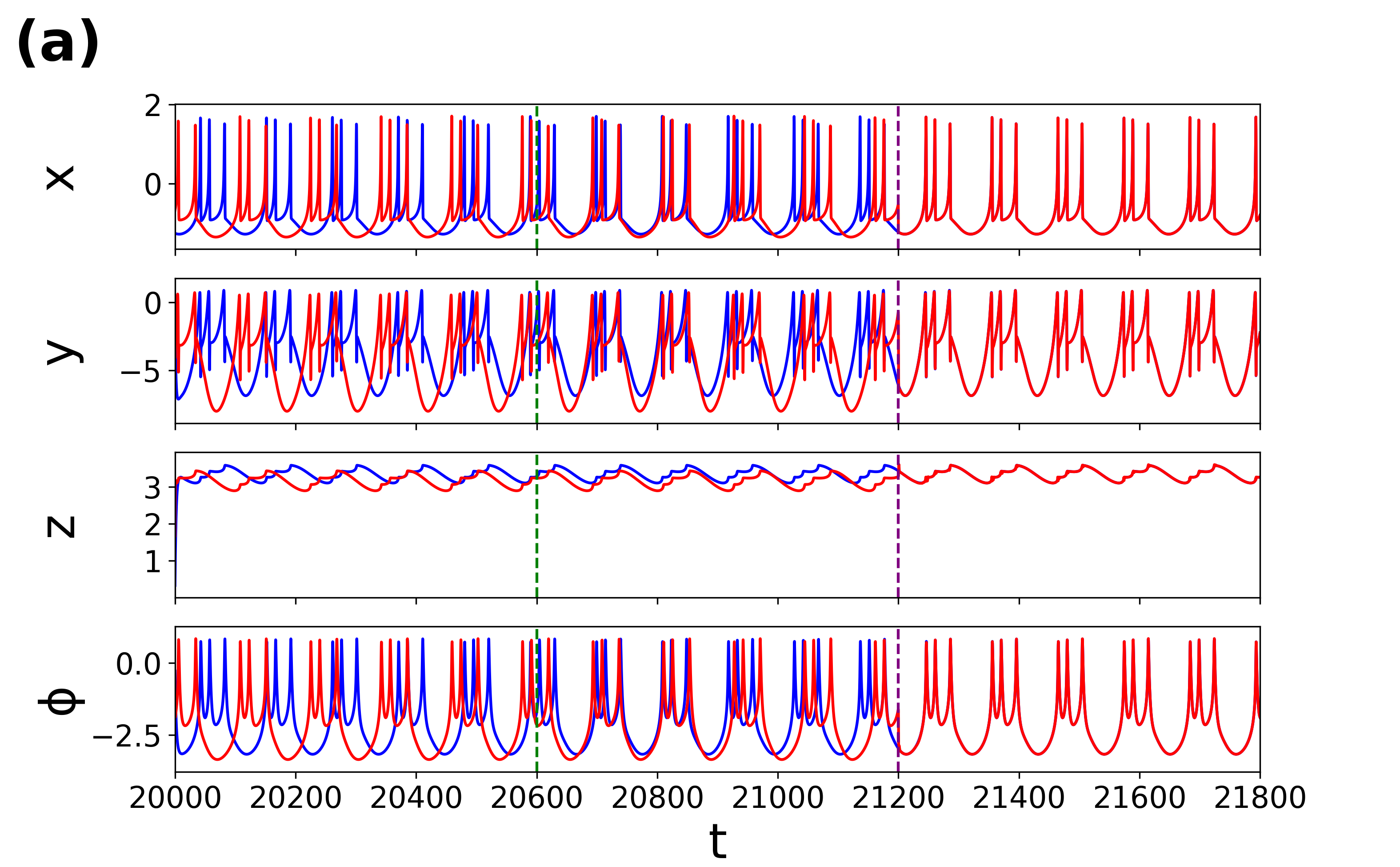}
        \includegraphics[width=12.0cm,height=6cm]{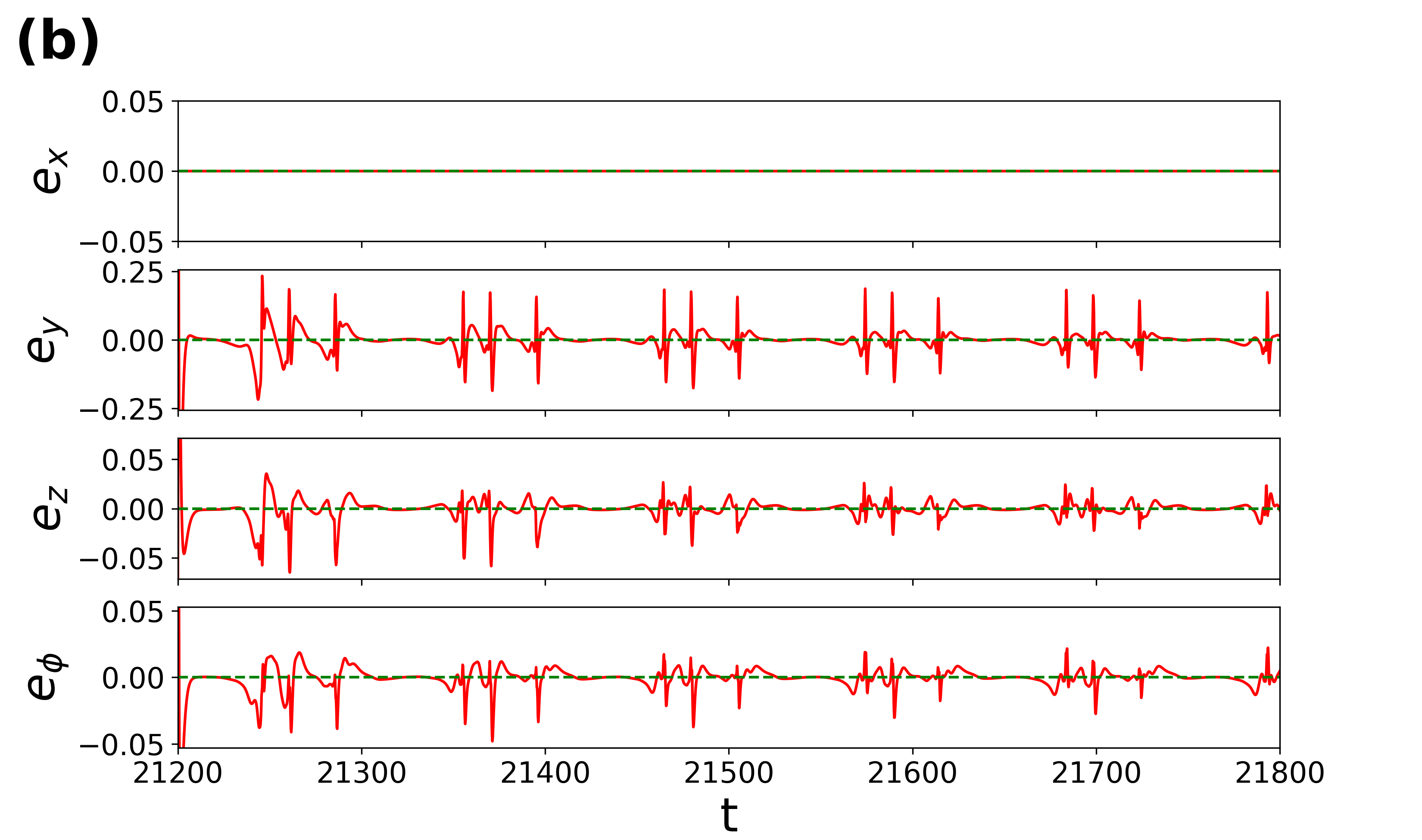}       
    \caption{Performance of the OC scheme synchronizing
ROs. \textbf{(a)} Time series of the drive RO (blue) and response RO (red). For training (before the vertical green line), inference (starts after the vertical green line), and control phases (starts after the vertical purple line).
    \textbf{(b)} Error between the drive and response ROs in the control phase.}\label{fig:esn_controller_1}
\end{figure}


\subsubsection{Online predictive control (OPC) scheme}

In the online control approach, we utilized the drive system's output $\hat{\mathbf{s}}(t)$ as input for the response system ${\mathbf{u}}_r(t+1)$. However, the online control approach is not optimal because the two systems are of different dimensions. We improve the performance of the online control scheme with a control scheme that resembles the model predictive control approach \cite{mpc}.

In our OPC scheme, we estimate the drive system's output $\hat{\mathbf{s}}(t+1)$ and call the estimate $\mathbf{k}_{est}(t+1)$. Our objective is to determine the input for the response system that yields $\hat{\mathbf{s}}_r(t+1) \approx \mathbf{k}_{est}(t+1)$. To obtain $\mathbf{k}_{est}(t+1)$, we need the observed variable $x(t+1)$, which is not available at the current time step $t$. We address this by training an additional ESN  to predict $\hat{x}(t+1)$ based on the current observed variable $x(t)$. We then use $\hat{x}(t+1)$ as a prediction for $x(t+1)$ in $\mathbf{u}_{est}(t+1)$, along with $\hat{\mathbf{s}}(t)$ and $\mathbf{r}(t)$ from the drive system, to solve Eqs. \eqref{eq:ESN} and \eqref{eq:ml_esn_out}. Figure \ref{fig:esn_transition} shows the error between the estimated $\mathbf{k}_{est}(t+1)$ and the dynamics of the drive ESN $\mathbf{\hat{s}}(t+1)$. See in Fig. \ref{fig:esn_transition} that the error bounds are very small, of the order of \num{e-4}. We then solve the following constrained optimization problem:

\begin{equation}
\begin{aligned}
& \underset{y_r(t),z_r(t),\phi_r(t)}{\text{min}}
& &\sum\big[RO_r(\hat{\mathbf{u}}_r(t+1)) - \hat{\mathbf{k}}_{est}(t+1)\big]^2, \\
& \text{subject to}
&&  y(t) - b  \leq y_r(t) \leq  y(t) + b, \\
&&& z(t) - b  \leq z_r(t)      \leq  z(t) + b,      \\
&&& \phi(t) - b \leq \phi_r(t)  \leq  \phi(t) + b,   \\
\end{aligned}
\end{equation}

where $RO_r$ refers to the response RO, $\hat{\mathbf{u}}_r(t+1) =[\hat{x}(t+1),y_r(t),z_r(t),\phi_r(t)]$, and $\hat{\mathbf{k}}_{est}(t+1) = [\hat{y}_k(t+1),\hat{z}_k(t+1),\hat{\phi}_k(t+1)]$. Since $\hat{x}(t+1)$ is an estimate, we regularize the solutions through bound constraints. We employ the
Limited-memory Broyden–Fletcher–Gold\-farb–Shan\-no algorithm \cite{L-BFGS-B} with bound constraints (L-BFGS-B) from the SciPy library \cite{2020SciPy-NMeth} to solve this optimization problem. The obtained solution is then used for the input of the response RO and concatenated with the newly observed variable $x(t+1)$ to form ${\mathbf{u}}_r(t+1)$.

\begin{figure}[H]
    \centering
        \includegraphics[width=12.0cm,height=8cm]{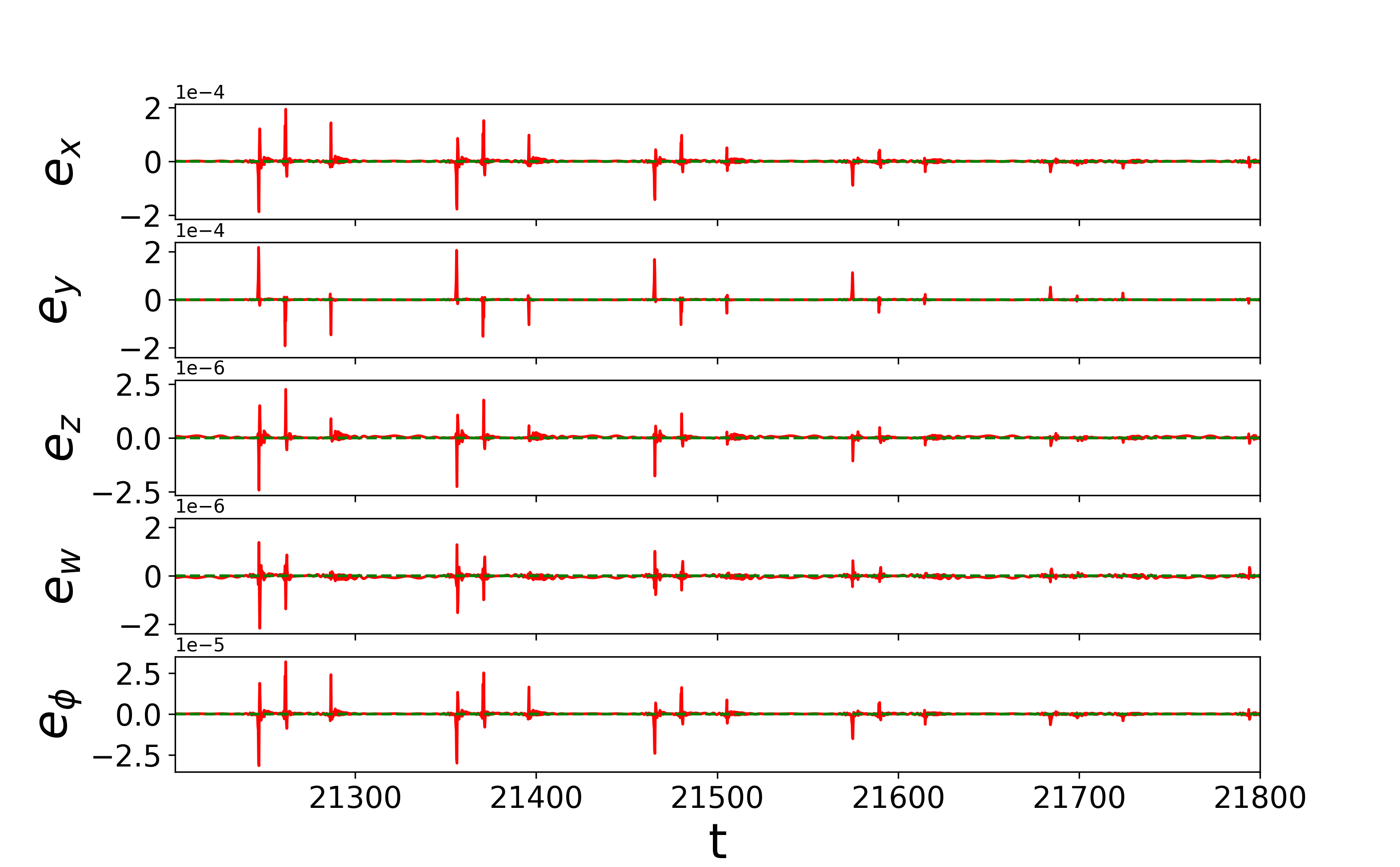}
    \caption{Error between $\hat{\mathbf{s}}(t+1)$ and $\mathbf{k}(t+1)$.}\label{fig:esn_transition}
\end{figure}

Figure \ref{fig:ML_c3_architecture} shows the complete architecture of our \textit{ad hoc} data-driven control algorithm for the reduced-order synchronization problem. Similarly to the online control scheme, we update $\tilde{\mathbf{W}}_{out}(t+1)$ by Eq. \eqref{ml:eq_online} as indicated by the green box. Yellow boxes indicate predicted values, which are discarded after the current iteration. The purple box represents the L-BFGS-B algorithm, which provides the input for the response RO.

\begin{figure}[H]
    \centering
    \includegraphics[width=17cm,height=13.0cm]{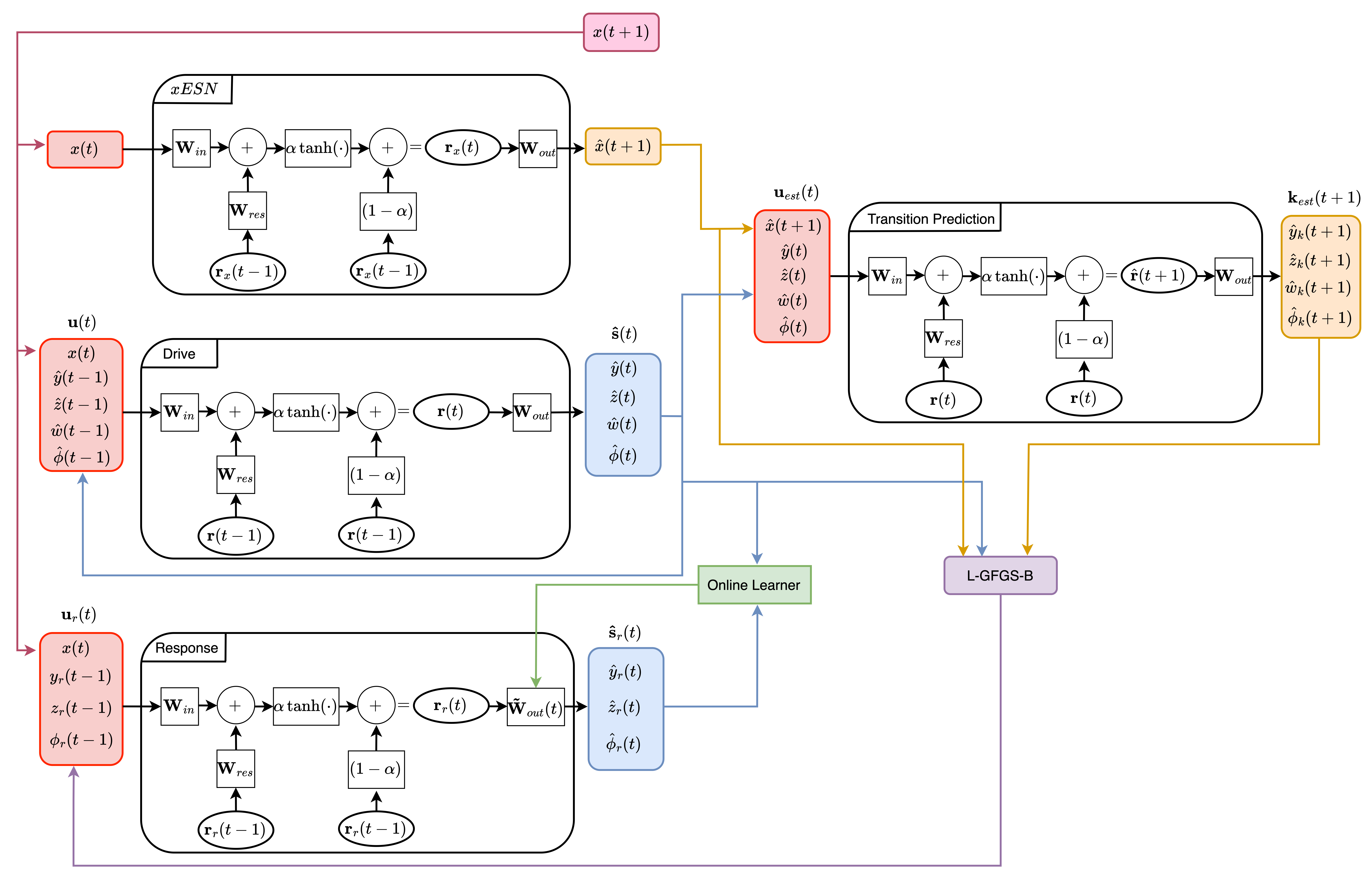}
    \caption{Block diagram illustrating the architecture during the control phase of the online predictive control scheme.}
    \label{fig:ML_c3_architecture}
\end{figure}

We set the bound parameter $b=0.5$ as it provided a good balance between low RMSE and a low error bound. Figure \ref{fig:esn_controller_2}\textbf{(a)} presents the time series of the drive and response ROs across three phases: training $[20000, 20600]$, prediction $[20601,21200]$, and control $[21201, 21800]$. The errors between the two ROs is illustrated in Fig. \ref{fig:esn_controller_2}\textbf{(b)}. The algorithm in Fig. \ref{fig:ML_c3_architecture} with a RMSE of \num{2.35e-2} outperforms the online control scheme in Fig. \ref{fig:esn_c2} with a RMSE of \num{3.70e-2}.

\begin{figure}[H]
    \centering
        \includegraphics[width=12.0cm,height=6cm]{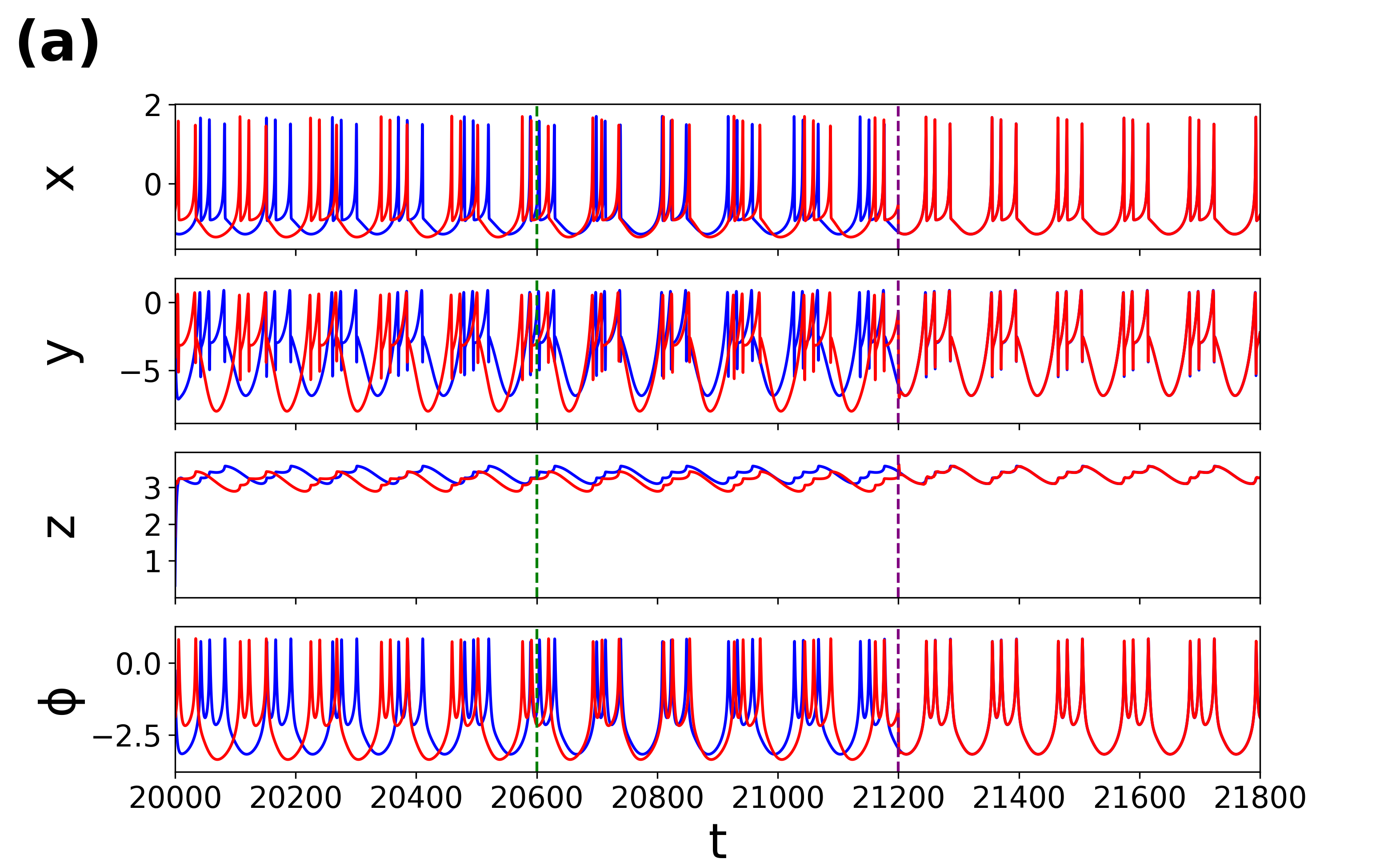}
        \includegraphics[width=12.0cm,height=6cm]{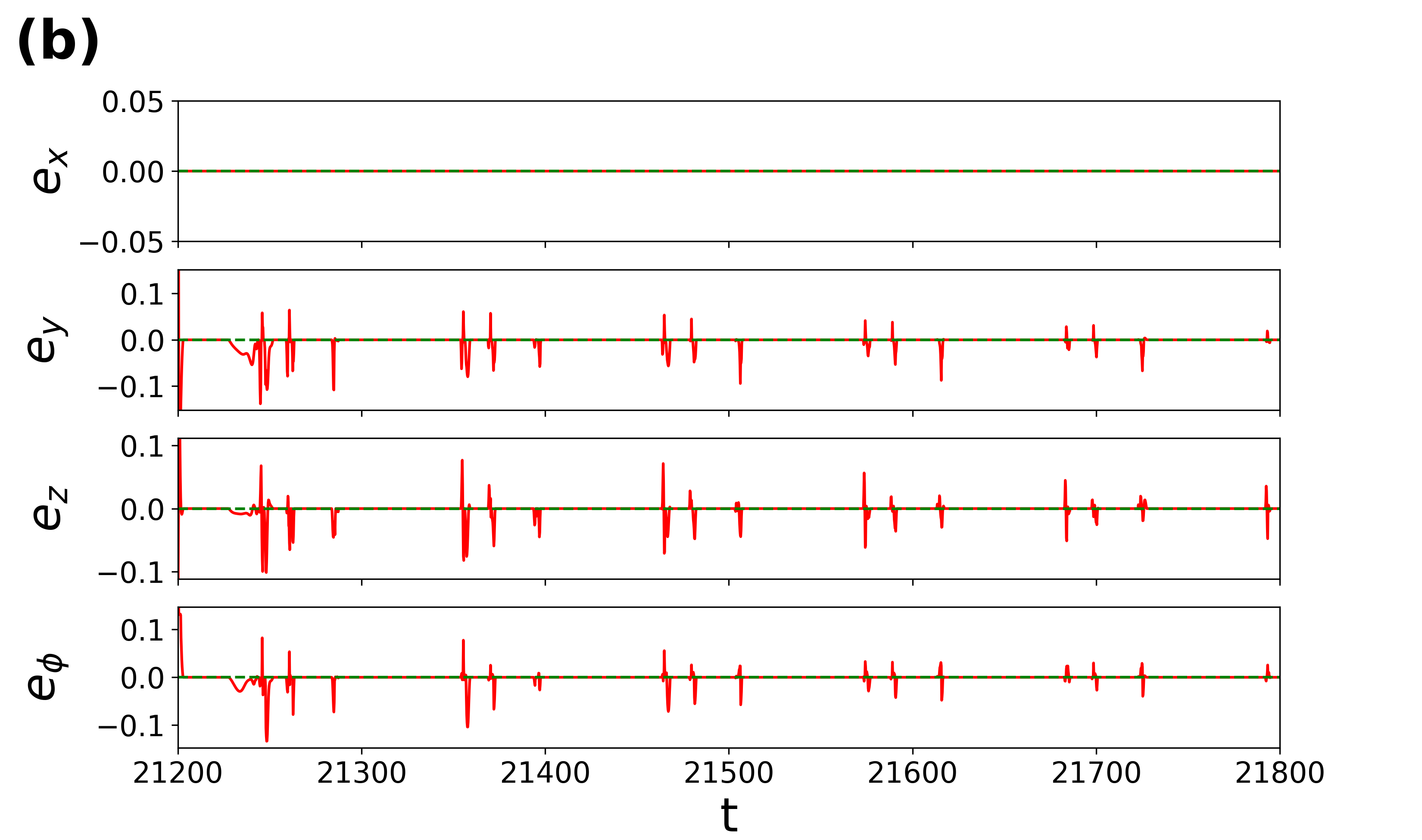}       
    \caption{Performance of the OPC scheme for synchronizing ROs. \textbf{(a)} Time series of the drive RO (blue) and response RO (red), showing three phases: training (before the green vertical line), prediction (between green and purple vertical lines), and control (after the purple vertical line)
(b) Synchronization error between the drive and response ROs in the control phase.}\label{fig:esn_controller_2}
\end{figure}

\section{Summary and conclusion}\label{sec:conclusion}

We investigated two distinct approaches to achieve re\-duc\-ed-order synchronization between the 4D and 5D HR neuron models using a dynamical systems approach (namely, an adaptive control scheme) and a machine learning approach. These methods differ significantly in both their methodology and application. The dynamical systems approach requires explicit mathematical representations of the underlying systems, enabling detailed theoretical analysis. Conversely, machine learning techniques are useful when explicit differential equations are unavailable relying instead on empirical measurements (data) collected from the system of interest.

In the dynamical systems approach to the reduced-order synchronization problem, we applied a Lyapunov-based adaptive control method by constructing a Lyapunov function to establish stability conditions of the reduced-order synchronization manifold. Numerical simulations demonstrated the effectiveness of our proposed adaptive control scheme. We achieved stable reduced-order synchronization across various firing modes (spiking and bursting) for periodic and chaotic time series.

In the machine learning approach to the reduced-order synchronization problem, we applied a combination of reservoir observer echo state networks, online learning, and online predictive learning schemes to achieve reduced-order synchronization.

The data-driven approach, unlike the dynamical systems approach, does not provide mathematical conditions for the stability of the synchronization manifold and offers limited insights into system behavior, particularly regarding the effects of varying bifurcation parameters on synchronization. However, in case of a lack of governing dynamical equations modeling a system the machine learning approach becomes indispensable.

Given that noise in ubiquitous brain neural networks \cite{noise_in_neurons}, an interesting future research direction on the topic would be to investigate reduced-order synchronization for stochastic neuron neural networks from both a dynamical system theory and machine learning perspective.

\section*{Acknowledgments}
This work was supported by the Department of Data Science (DDS), Friedrich-Alexander-Universit\"at Erlan\-gen-Nürnberg, Germany.

\section*{Data and code availability statement} 
The simulation data supporting the findings of this study, along with the code used to produce the results, are publicly accessible here~\cite{code}.

\section*{Declaration of competing interest}
The authors declare that there is no conflict of interest with this article.
%
%

\end{document}